\documentclass{hha}
\usepackage{graphicx, bbold, enumerate, mathrsfs, euscript}
\usepackage[titletoc,title]{appendix}
\usepackage[all]{xy}
\usepackage[retainorgcmds]{IEEEtrantools}
\usepackage{xspace}
\usepackage[english]{babel}
 
\theoremstyle{hhaplain}

\newif\ifmaths
\newcommand{\mathspacing}[1]{
	\relax\ifmmode\mathstrue\else\mathsfalse\fi
	{\ifmaths#1\else$\!\!#1$\fi}\xspace}

\newcommand{\vect}{\protect\mathspacing{\operatorname{Vect}_{\mathbb{C}}}}
\newcommand{\vects}{\protect\mathspacing{\operatorname{Vect}_{\mathbb{C}}^s}}

\newcommand{\catoc}{\protect\mathspacing{\mathcal{K}}}
\newcommand{\catc}{\protect\mathspacing{\bar{\mathcal{N}}}}
\newcommand{\cat}{\protect\mathspacing{\mathcal{N}}}
\newcommand{\catop}{\protect\mathspacing{\bar{\mathcal{O}}}}
\newcommand{\catoo}{\protect\mathspacing{\mathcal{O}}}

\newcommand{\lcatoc}{\mathspacing{\mathfrak{K}}}
\newcommand{\lcatc}{\mathspacing{\bar{\mathfrak{N}}}}
\newcommand{\lcat}{\protect\mathspacing{\mathfrak{N}}}

\newcommand{\lcatoo}{\mathspacing{\mathfrak{O}}}

\newcommand{\lcatocone}{\mathspacing{\mathfrak{K}_{(1,0)}}}

\newcommand{\lcato}{\mathspacing{\mathfrak{N}_0}}
\newcommand{\lcatone}{\mathspacing{\mathfrak{N}_1}}
\newcommand{\lcatn}{\mathspacing{\mathfrak{N}_n}}
\newcommand{\lcatm}{\mathspacing{\mathfrak{N}_m}}

\newcommand{\lcatooo}{\mathspacing{\mathfrak{O}_0}}
\newcommand{\lcatooone}{\mathspacing{\mathfrak{O}_1}}
\newcommand{\lcatootwo}{\mathspacing{\mathfrak{O}_2}}
\newcommand{\lcatoon}{\mathspacing{\mathfrak{O}_n}}

\newcommand{\cato}{\mathspacing{\mathcal{N}_0}}

\newcommand{\catb}{\mathspacing{\mathcal{N}_b}}

\newcommand{\catone}{\mathspacing{\mathcal{N}_1}}
\newcommand{\catoneor}{\mathspacing{\mathcal{N}^{\cator}_1}}
\newcommand{\catonenor}{\mathspacing{\mathcal{N}^{\catnor}_1}}

\newcommand{\nn}{\ensuremath{\mathbb{N}\times\mathbb{N}}}
\newcommand{\sminfty}{\hbox{\hspace{-2pt}\raise0pt\hbox{$_{^\infty}$}}}
\newcommand{\cator}{\hbox{\raise2pt\hbox{$_{+}$}}}
\newcommand{\catnor}{\hbox{\raise2pt\hbox{$_{-}$}}}

\newcommand*{\defeq}{\mathrel{\vcenter{\baselineskip0.5ex \lineskiplimit0pt\hbox{\scriptsize.}\hbox{\scriptsize.}}}=}
\newcommand{\smallotimes}{\hbox{\raise2pt\hbox{$\,_\otimes\,$}}}

\begin{document}

\title{Localisations of cobordism categories and invertible TFTs in dimension two}
\author{R. Juer}
\email{juerr@maths.ox.ac.uk}
\address{Mathematical Institute, 
	Oxford University, 
	OX1 3LB, 
	UK}
\author{U. Tillmann}
\email{tillmann@maths.ox.ac.uk}
\address{Mathematical Institute,
	Oxford University,
	OX1 3LB,
	UK}


%
\begin{abstract}
Cobordism categories have played an important role in classical geometry and more recently in mathematical treatments of quantum field theory. Here we will compute  localisations of  two-dimensional discrete cobordism categories. This allows us, up to equivalence, to determine the category of invertible two-dimensional topological field theories in the sense of Atiyah. We are able to treat the orientable, non-orientable, closed and open cases.   
\end{abstract}

\received{Month Day, Year}   
\revised{Month Day, Year}    
\published{Month Day, Year}  
\submitted{Chuck Weibel}      
\volumeyear{2013} 
\volumenumber{??} 
\issuenumber{1}   
\startpage{1}     
\articlenumber{1} 
\owner{International Press}

\maketitle
%
%
\section{Introduction}
\label{intro}
%
	The mid 1980s saw a shift in the nature of the relationship between mathematics and physics. Differential equations and geometry applied in a classical setting were no longer the principal players; in the quantum world topology and algebra began to move to the fore. This was a result of extracting the topological aspects of a quantum field theory, largely attributed in its infancy to Witten, for example in \cite{W}. The move to a formal mathematical theory came in 1989 with Atiyah's axiomatisation of a topological quantum field theory (TFT) \cite{A}. These ideas have since shown themselves to have applications to both topology (low dimensional manifolds) and physics (string theory) alike.
	
	Roughly speaking, a closed $2$-dimensional TFT is a functor from a category whose objects are closed oriented $1$-manifolds (disjoint unions of circles), and whose morphisms are homeomorphism classes of oriented surfaces, to the category of complex vector spaces. It is well known that closed $2$-dimensional TFTs are in one-to-one correspondence with commutative Frobenius algebras. Explicitly, the image of the circle $S^1$ under such a theory is an algebra of this form and, conversely, for any such algebra $A$ there exists a theory $F$ with $F(S^1)=A$. Proofs of this folk theorem have been given by Abrams \cite{Ab} and Kock \cite{K}.
	
	Atiyah's axiomatisation can be extended in two ways. Firstly, we can drop the assumption that all objects ($1$-manifolds) in the category are closed, requiring them only to be compact. This leads to the notion of \emph{open-closed} and \emph{open} TFTs, where for the latter we insist that the objects have boundary. Theories of this form have also been classified. Moore and Segal \cite{MS} showed that open $2$-dimensional TFTs are equivalent to symmetric Frobenius algebras. In the open-closed case there is a one-to-one correspondence with ``knowledgeable Frobenius algebras'' (see Lauda and Pfeiffer \cite{LP} for a definition and proof).
	
	The second way in which we can extend the axiomatisation is by allowing non-orientable surfaces in the category. Field theories of this form are known as \emph{Klein} TFTs. Turaev and Turner \cite{TT} proved that closed $2$-dimensional Klein TFTs correspond to finite dimensional commutative Frobenius algebras with additional structure coming from the non-orientable generators of the surface category. Open Klein TFTs are equivalent to symmetric Frobenius algebras with some extra structure, as shown by Braun \cite{B}. Finally, open-closed Klein theories correspond to ``structure algebras'' (see Alexeevski and Natanzon \cite{AN} for a definition and proof).
	
	Here we are interested in invertible $2$-dimensional TFTs, that is to say TFTs for which the defining functor $F$ is morphism inverting. These play an important role in ordinary quantum field theory, for example as anomalies, and over recent years have proven crucial in the understanding  of the subject. For orientable surfaces invertible TFTs have been classified, in the closed case by the second author \cite{T} and in the open-closed case by Douglas \cite{D}. Our aim is to extend these results to the corresponding non-orientable cases, and also to give a classification in the open case which has not previously been considered. 

	The previous classifications have been dependent on a particular choice of model for the cobordism categories considered. Though this is convenient for computations, it is unsatisfactory in terms of the general theory. Here we make our classification independent of such a choice and determine the category of TFTs on any model for the cobordism categories, up to natural equivalence. To do so we are forced to consider the category of (pointed) symmetric monoidal functors and prove that the equivalence class of such a functor category does not change when the source or target categories are replaced by equivalent symmetric monoidal categories.    
	
	Our interest in TFT lies not only with the functors themselves, but also with the surface categories involved. By studying their localisations, i.e.\ the categories obtained by inverting all morphisms, one can gain an understanding of the underlying combinatorics of these categories.  Here we adapt the methods of \cite {T} in order to analyse the localisation of the non-orientable cobordism categories. We are able to compute them in all three cases, obtaining the desired classification of invertible field theories as a corollary. For completeness' sake we have also included proofs for the orientable cobordism categories.		

	The categorical problem at hand is the discrete version of an extended problem, in which one includes diffeomorphism of surfaces in the categorical data. In \cite{GTMW} Galatius, Madsen, the second author and Weiss showed that there is a weak equivalence between the classifying space of the extended $d$-dimensional cobordism category  $\EuScript{C}_d$ and the infinite loop space of a certain Thom spectrum. This result can be interpreted to calculate the extended invertible TFTs in the setting of \cite {Lurie}. We refer to the recent article  by Freed \cite {F} for a survey of the discrete and extended cobordism categories and the TFTs they define, along with various examples and further applications.

	\noindent
	{\bf  Outline and results:}

	We will mainly study the following   cobordism categories in dimension two.

	\begin{description}
		\item [\catoc ] {\it    open-closed category:}
		{\/Objects are unoriented compact $1$-manifolds (disjoint unions of circles and intervals). A morphism $\Sigma:M_0\rightarrow M_1$ is a (not necessarily orientable) compact cobordism from $M_0$ to $M_1$ up to homeomorphism. We call $\partial\Sigma\smallsetminus\left(M_0\cup M_1\right)$ the \emph{free} boundary of $\Sigma$.}
		\item [\cat ] {\it closed category:} 
		{\/The subcategory of \catoc whose objects are disjoint unions of circles, and
		whose morphisms have no free boundary.}	
		\item [\catoo ]{\it open category:}    
		{\/The subcategory of \catoc whose objects are disjoint unions of intervals, and
		whose morphisms have no closed components.}
		\item [$\mathcal S$] {\it orientable category:}
		{\/The subcategory of \catoc in which only orientable surfaces are morphisms.}
	\end{description}

	In section 2,  we define carefully the two-dimensional cobordism categories \catoc, \cat, $\mathcal S$ and several other  subcategories. Using the Euler characteristic we define a symmetric monoidal functor $\Theta : \catoc \rightarrow \mathbb Z$. 

	Section 3 is the heart of the paper. Here we compute the localisations of all categories that we consider. It is shown that $\Theta$ induces on the localisations of \cat,   \catoo and \catoc an equivalence with $\mathbb Z$, see Theorems 3.6, 3.8, and 3.11. The same is true when considering the subcategories with only  orientable surfaces. In the case of \cat and $\cat \cap \mathcal S$ this gives in particular an elementary computation of the fundamental groups $\pi_1 (B \EuScript C_2) \simeq \mathbb Z$ and $\pi_1( B\EuScript C_2^+) \simeq \mathbb Z$ in the notation of \cite {GTMW}.  

	In section 4, we apply these results in order to describe the  category of TFTs (with target the category of complex vector spaces). In particular, we show that when the domain is \cat , \catoo, \catoc or $\mathcal S$ the  category of invertible TFTs is equivalent to the discrete category of non-zero complex numbers, see  Theorem 4.3.

	In section 5  we consider the classifying spaces of our cobordism categories and relate our results to those in \cite {GTMW} and \cite {Lurie}. Our results compute the invertible (discrete) TFTs in the $(\infty ,1)$-setting  if and only if Conjecture 5.3 holds.
 
	Appendix A contains the categorical definitions and results on which section 4 relies. It is a well-known fact that when the source  or target category are replaced by equivalent categories then the resulting functor categories are equivalent. The main purpose of the appendix is to show that this remains true when one considers pointed symmetric monoidal functors and when in turn one restricts to invertible functors. Though not surprising these results are not readily available in the literature. Because of their key role in section 4 we have therefore included a detailed treatment. 

	\ack
	We  would like to thank Chris Douglas and Jeff Giansiracusa for their interest and comments  as this project developed. 

%
%

%
%
%
%
%
%
\section{Cobordism categories and the functor $\Theta$}
\label{definitions}

	We introduce a convenient model of the open-closed cobordism category  
	and its subcategories in detail.
	A modification of the Euler characteristic  defines a functor $\Theta$ 
	to the integers. This functor plays a central role in the following 
	sections. 

	\subsection{The open-closed category \catoc}
		The objects of \,\catoc are compact unoriented $1$-manifolds. Morphisms are (not necessarily orientable) cobordisms up to homeomorphism relative to the boundary. Explicitly, a morphism $(\Sigma,\sigma_0,\sigma_1):M_0\rightarrow M_1$ in \catoc is a compact $2$-manifold $\Sigma$ with maps $\sigma_0: M_0\rightarrow\partial \Sigma$ and $\sigma_1:M_1\rightarrow\partial \Sigma$ which are homeomorphisms onto their images and satisfy $\sigma_0(M_0)\cap\sigma_1(M_1)=\emptyset$. The morphism $(\Sigma,\sigma_0,\sigma_1)$ will often be abbreviated as $\Sigma$. We call $\sigma_0(M_0)$, $\sigma_1(M_1)$ and $\partial\Sigma\smallsetminus(\sigma_0(M_0)\cup\sigma_1(M_1))$ the \emph{source boundary}, the \emph{target boundary} and the \emph{free boundary} respectively. We say two cobordisms $\Sigma, \Sigma':M_0\rightarrow M_1$ are homeomorphic relative to the boundary (and therefore define the same morphism in \catoc) if there exists a homeomorphism $\Psi:\Sigma\rightarrow\Sigma'$ such that the following diagram commutes.
%
		\begin{displaymath}
			   \xymatrix{
					& M_0 \ar@/_0.7pc/[dl]_{\sigma_0} \ar@/^0.7pc/[dr]^{\sigma_0'} & \\
		    		\Sigma \ar[rr]^{\Psi} & & \Sigma' \\
		     		& M_1 \ar@/^0.7pc/[ul]^{\sigma_1} \ar@/_0.7pc/[ur]_{\sigma_1'} & \\
			}
		\end{displaymath}
		Composition of $\Sigma:M_0\rightarrow M_1$ and $\Sigma':M_1\rightarrow M_2$ is given by glueing the two cobordisms via $\sigma'_0\circ\sigma_1^{-1}$, and has boundary maps $\sigma_0$ and $\sigma'_1$. The identity morphism from $M_0$ to itself is given by the cylinder $M_0\times I$.  \catoc is a symmetric  monoidal category under disjoint union of manifolds.
			
		It is convenient to identify \catoc with a skeleton. Let $M$ be a $1$-manifold in \catoc with $m$ closed components and $n$ components with boundary. For each closed component $M_i$  ($1\leq i\leq m$) we choose a homeomorphism $f_i:M_i\rightarrow S^1$, where $S^1$ denotes the standard (unoriented) circle. Similarly for each component $N_i$ ($1\leq i\leq n$) with boundary we choose a homeomorphism $g_i:N_i\rightarrow I$, where $I=[0,1]$ denotes the standard (unoriented) interval. Denote the mapping cylinders of $f_i$ and $g_i$ by $C_{f_i}$ and $C_{g_i}$ respectively. Then the union of mapping cylinders $\left(\coprod_{i=1}^{m}C_{f_i}\right)\amalg\left(\coprod_{i=1}^{n}C_{g_i}\right)$ defines an isomorphism $M\rightarrow\left(\coprod_{i=1}^mS^1\right)\amalg\left(\coprod_{i=1}^nI\right)$ in \catoc. Therefore, the full subcategory on disjoint unions of copies of a single unoriented circle $S^1$, and a single unoriented  interval $I$, is a skeleton for \catoc. We identify these objects with the product \nn, where $(0,0)$ represents the empty $1$-manifold and $(m,n)$ represents $m$ ordered copies of $S^1$ and $n$ ordered copies of $I$, and we shall often refer to an object of \catoc in this manner.
			
		For each morphism in \catoc, we choose a representative cobordism whose source and target boundary components are identified with copies of $S^1$ and $I$, and whose boundary maps are either inclusions or reflections on each component of an object $(m,n)$. Note that the boundary maps are essential in determining the homeomorphism class of a cobordism in \catoc. For example, consider the cylinder as a morphism $S^1\rightarrow S^1$. If both boundary maps are defined to be inclusions, then this is just the identity $S^1\times I$. However, if we define one boundary map to be an inclusion, and the other to be a reflection, we obtain an entirely different morphism. Although the underlying manifolds of the two cobordisms are homeomorphic, we cannot find a homeomorphism which commutes with their boundary maps. The latter cobordism can be thought of as the mapping cylinder of the reflection $r:S^1\rightarrow S^1$ and, with this in mind, we denote it by $C_r^{S^1}$. Similarly, the disc as a morphism $I\rightarrow I$ comes in two guises; the identity, and the mapping cylinder of the reflection $r:I\rightarrow I$. We shall often refer to the latter as the \emph{bow tie}, denoting it by $C_r^I$.
			
		The orientable symmetric monoidal subcategory $\mathcal{S}$ of \catoc consists of oriented $1$-manifolds and oriented cobordisms whose boundary maps are orientation preserving.\footnote{We identify $\mathcal{S}$  with its  skeleton:  the full subcategory on unions of copies of one fixed oriented circle and one fixed oriented interval. In this way we view $\mathcal{S}$ as a subcategory of \catoc by forgetting orientations.} A generators and relations description of $\mathcal{S}$ has been given by Lauda and Pfeiffer \cite{LP}. The fourteen cobordisms in Figure~\ref{GeneratorsOr} form a generating set for $\mathcal{S}$ under the operations of composition and disjoint union. In the figure, source boundary components are on the left, and target boundaries are on the right. For those cobordisms with free boundary, the thickened lines are intended to distinguish source and target boundary components from free components. For example, the bottom right hand cobordism is a morphism from $S^1$ to $I$.

		\begin{figure}[htp]
			\begin{center}
			\includegraphics[scale=0.56]{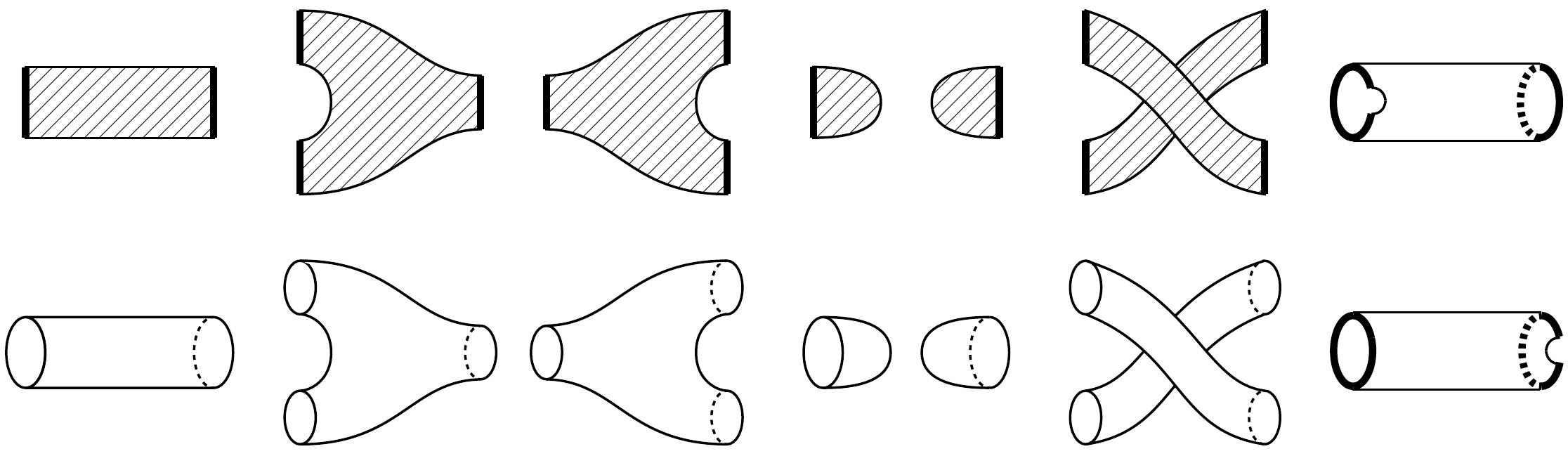}
			\caption{Generators of $\mathcal {S}$.}\label{GeneratorsOr}
			\end{center}
		\end{figure}
		By moving crosscaps and taking out factors of $C_r^{S^1}$ and $C_r^I$, we can decompose any cobordism in \catoc as an element of $\mathcal{S}$ composed with a combination of the four additional cobordisms in Figure~\ref{GeneratorsUnor}. In the figure, a dotted circle surrounding a cross represents a crosscap attached to a surface. From left to right, the cobordisms are: the projective plane with two discs removed, the M\"obius strip, the cylinder $C_r^{S^1}$, and the disc $C_r^I$.

		\begin{figure}[htp]
			\begin{center}	
			\includegraphics[scale=0.56]{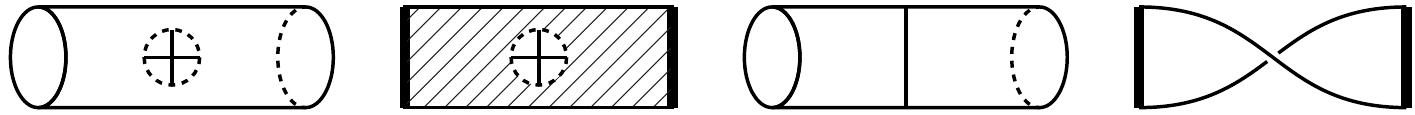}
			\caption{Additional generators of \catoc.}\label{GeneratorsUnor}
			\end{center}
		\end{figure}
		It follows that the eighteen cobordisms of Figures~\ref{GeneratorsOr} and~\ref{GeneratorsUnor} form a generating set for \catoc under composition and disjoint union. 
	\subsection{The category \cat and  and its subcategories}

		We define the closed category \cat to be the following symmetric monoidal subcategory of \catoc. The objects of \cat are all closed $1$-manifolds in \catoc, and morphisms $\Sigma:M_0\rightarrow M_1$ must have boundary maps satisfying the additional condition $\sigma_0(M_0)\cup\sigma_1(M_1)=\partial\Sigma$. In other words, we do not allow cobordisms in \cat to have any free boundary. 

		We will need a careful analysis of various subcategories of $\cat$.
		\begin{displaymath}
			\xymatrix@R=5pt@C=50pt{
				& & \;\cato\; \ar@{^(->}[dr] & \\
				\;\catoneor\;\ar@{^(->}[dr] & & & \;\cat\; \\
				& \;\catone\;\ar@{^(->}[r] & \;\catb\;\ar@{^(->}[ur] \\
				\;\catonenor\;\ar@{^(->}[ur] & & & &
			}
		\end{displaymath}

		\begin{definition}
		Let \catone be the subcategory of \cat consisting of connected endomorphisms of the circle $S^1$.
		\end{definition}
		Morphisms in \catone are connected surfaces with one source and one target boundary circle, and are determined by their genus or number of crosscaps and their boundary maps. Note that composition in \catone is abelian.
	Recall from Section 2.1 that for each morphism we choose a representative cobordism whose boundary circles are identified with $S^1$, and whose boundary maps are either inclusions or reflections.
		\begin{definition}
			Let $(\Sigma,\sigma_0,\sigma_1)$ be a cobordism of the chosen form representing a morphism in \catone. If $\sigma_0$ and $\sigma_1$ are both inclusions, or both reflections, then we say that they have the \emph{same direction}. In this case we say that $\Sigma$ is of \emph{type 0}. If one of $\sigma_0$ and $\sigma_1$ is an inclusion, and the other is a reflection, then we say that the boundary maps have \emph{opposite directions}. In this case we say that $\Sigma$ is of \emph{type 1}.
		\end{definition}
		For orientable surfaces,\footnote{Note that by orientable we do not mean oriented, and so no assumption is made about the direction of boundary maps.} the type of a cobordism carries through to a well-defined notion of the type of a morphism. To see this, we make the following observations. Firstly, for a cobordism of genus $g$ and type $0$, the two possible choices (boundary maps are both inclusions or boundary maps are both reflections) lie in the same homeomorphism class. Similarly, for a cobordism of genus $g$ and type $1$ the two possible choices define the same morphism. Finally, there is no homeomorphism between an orientable cobordism of genus $g$ and type $0$, and an orientable cobordism of genus $g$ and type $1$, which commutes with all boundary maps. Thus the type of an orientable morphism is well-defined. For example, the cylinder of type $1$ in \catone is the morphism $C_r^{S^1}$, whilst the cylinder of type $0$ is the identity.

		Now consider the subcategory $\catoneor$ of \,\catone whose morphisms are all orientable. Morphisms in $\catoneor$ are determined by their genera and their type. In terms of genera, composition in $\catoneor$ corresponds to addition. If two morphisms in $\catoneor$ have the same type, then their composition will always have type $0$. However, composing two morphisms of different types gives a morphism of type $1$. It follows that the type of a morphism lies in the group $\mathbb{Z}_2$. We therefore identify
		\[
			\catoneor = \mathbb{N}\times\mathbb{Z}_2,
		\]
		where the element $(g,\epsilon)$ stands for the unique morphism in $\catoneor$ of genus $g$ and type $\epsilon\in\{0,1\}$.

		Similarly, we can consider the subcategory $\catonenor$ of \catone whose morphisms are non-orientable cobordisms along with the cylinder. For non-orientable surfaces, the type of a cobordism does not carry through to a well-defined notion of morphism type. This is precisely because a cobordism from $S^1$ to $S^1$ with $k\geq 1$ crosscaps and type $0$ lies in the same homeomorphism class as a cobordism from $S^1$ to $S^1$ with $k$ crosscaps and type $1$. To see this, let $(\Sigma,\sigma_0,\sigma_1)$ be the  M\"obius band with one disc removed considered  as a cobordism from $S^1$ to $S^1$ of type $1$. This is depicted in the left of Figure~\ref{MobiusSlide}, where the boundary circle of the M\"obius band represents the image of $\sigma_0$, and the boundary of the removed disc represents the image of $\sigma_1$. The arrows denote the directions of the embeddings $\sigma_0$ and $\sigma_1$. The surface on the right of Figure~\ref{MobiusSlide} is the punctured M\"obius band $\Sigma'$ of type $0$. The required homeomorphism $\Sigma\rightarrow\Sigma'$ is then given by the crosscap slide, that is by pushing the image of $\sigma_1$ once round the M\"obius strip, and this commutes with all boundary maps.
		\begin{figure}[htp]
			\begin{center}
				\includegraphics[scale=0.55]{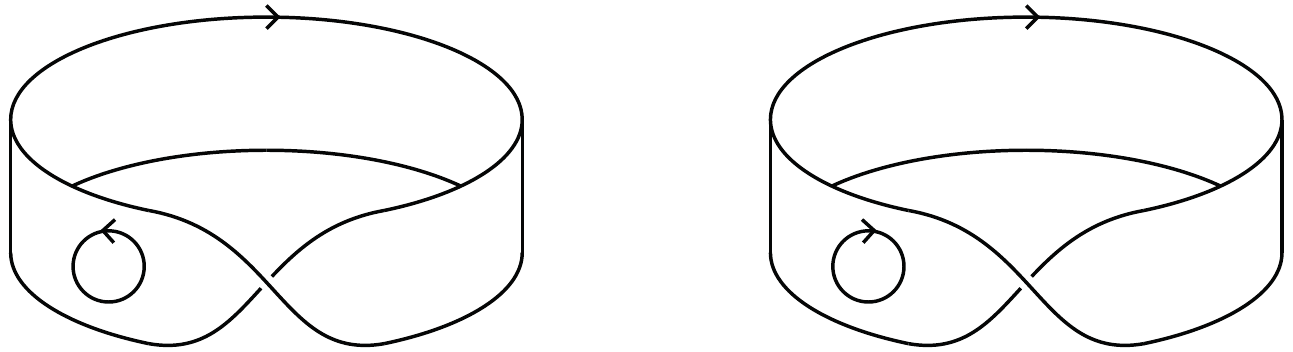}
				\caption{M\"obius bands of type $1$ and of type $0$.}\label{MobiusSlide}
			\end{center}
		\end{figure}
		\linebreak Thus for each $k\geq 1$ there is only one morphism with $k$ crosscaps in $\catonenor$. In terms of crosscaps, composition in $\catonenor$ corresponds to addition. It follows that $\catonenor$ is identifiable with the monoid $\mathbb{N}\cup C_r^{S^1}$, where $C_r^{S^1}$ acts as an identity for all elements other than itself and the cylinder $S^1\times I$.
		Using  notation compatible with that used  for $\catoneor$, we can identify
		\[
			\catonenor=\frac{\mathbb{N}\times\mathbb{Z}_2}{\sim},
		\]
		where the relation is $(k,0)\sim(k,1)$ for non-zero $k$.

		\begin{proposition}
			The monoid \catone can be identified with  $\frac{\nn\times\mathbb{Z}_2}{\sim}$ where the relation is generated by $(g,k,0)\sim(0,2g+k,1)$ if and only if $k$ is non-zero.
		\end {proposition}
		\begin{proof}
	        Piecing together the above results, we see that \catone can be identified with the monoid $\nn\times\mathbb{Z}_2$ modulo the relations discussed above, where the class $(g,k,\epsilon)$ is represented by a cobordism with $g$ handles, $k$ crosscaps, and type $\epsilon\in\{0,1\}$. One checks that the single relation $(g,k,0)\sim(0,2g+k,1)\iff k\neq 0$ generates all others, and hence $\catone=\frac{\nn\times\mathbb{Z}_2}{\sim}$ where the relation is $(g,k,0)\sim(0,2g+k,1)$ for non-zero $k$.
		\end{proof}

	\subsection{The functor $\Theta:\catoc\rightarrow\mathbb{Z}$}
%
		It is well known that any connected closed surface is homeomorphic to one and only one of the following: the sphere $S^2$, a connect sum of $g$ tori for $g\geq1$, or a connect sum of $k$ real projective planes for $k\geq1$. The Euler characteristic $\chi$ of such a surface can be computed by $2-2g$ in the orientable case, and $2-k$ in the non-orientable case. Removing $n$ discs from the surface reduces its Euler characteristic by $n$. In general, for spaces $X_1$ and $X_2$ whose union is $X$, we have $\chi(X)=\chi(X_1)+\chi(X_2)-\chi(X_1\cap X_2)$. Consider $\mathbb {Z}$ as a category with one object and endomorphism set $\mathbb {Z}$. We define a functor $\Theta\colon \catoc\rightarrow\mathbb{Z}$ which sends all objects of \catoc to the only object of $\mathbb{Z}$, and sends a morphism $\Sigma\colon (m,n)\rightarrow (p,q)$ in \catoc to
		\[
			\Theta(\Sigma)\defeq (p+q)-m -\chi(\Sigma).
		\]
%
		\begin{proposition}
			$\Theta\colon \catoc\rightarrow\mathbb{Z}$ is a functor of symmetric monoidal categories.
		\end{proposition}
		\begin{proof}
			Let $\Sigma_1\colon (p,q)\rightarrow (r,s)$ and $\Sigma_2\colon (m,n)\rightarrow(p,q)$ be morphisms in \catoc. Then, since the Euler characteristics of the circle and the unit interval are $0$ and $1$ respectively, we have
		\begin{align*}
			\Theta(\Sigma_1\circ\Sigma_2) & = (r+s)-m -\chi(\Sigma_1\circ\Sigma_2)\\
			& (r+s)-m - \left(\chi(\Sigma_1)+\chi(\Sigma_2)-q\right)\\
			& =(r+s) -p -\chi(\Sigma_1)+(p+q)-m- \chi(\Sigma_2)\\
			& = \Theta(\Sigma_1)+\Theta(\Sigma_2).
		\end{align*}
			$\Theta$ is additive on disjoint unions and hence is a map of symmetric monoidal categories.
		\end{proof}
%
%
		The functor $\Theta$ describes the category of surfaces surprisingly well. Our main results in Section 3 show that $\Theta$ defines  equivalences on the localisation of \catoc and  certain subcategories. Furthermore,  we will  now prove that $\Theta$ restricted to a large subcategory  has a right  adjoint. 

		\begin{definition}
			We define the category \catb to be the subcategory of \cat containing all objects of \cat, and those morphisms no connected component of which is a cobordism to zero.
		\end{definition}
		For example, a disc as a morphism $0\rightarrow 1$ would be in \catb, but a disc as a morphism $1\rightarrow 0$ would not. 

		\begin{theorem}
			The functor $\Theta$ when restricted to \catb has a right inverse which is also a right adjoint.
		\end{theorem}

		\begin{proof}
			First note that $\Theta$ restricted to $\catb$ takes only non-negative values, i.e. $\Theta\colon \catb \rightarrow \mathbb{N}$.
			Define  $i\colon \mathbb {N} \rightarrow \catb$ to be the inclusion that maps the only object to the circle and the morphism given by $k$ to the connected sum of $k$ projective planes, i.e.\ the non-orientable surface of genus $k$ with two discs removed.   
			Note that $\Theta\circ i=\operatorname{Id}_{\mathbb{N}}$. Therefore, to prove that the inclusion $i$ is right  adjoint to $\Theta$ it suffices to define a natural transformation $\tau\colon\operatorname{Id}_{\catb}\rightarrow i\circ\Theta$. For the object $n$ in \catb, let $\tau_n\colon n\rightarrow 1$ be the pair of pants surface with $n$ legs and one crosscap. For a morphism $\Sigma_{g,k,c}\colon n\rightarrow m$ with $c$ components, $k$ crosscaps, and total genus $g$, we need to check that the following diagram commutes.
		\begin{displaymath}
			\xymatrix{
				n \ar[d]_{\Sigma_{g,k,c}} \ar[r]^{\tau_n} & 1 \ar[d]^{2g+k+2m-2c} \\
				m \ar[r]_{\tau_m} & 1 \\
			}
		\end{displaymath}
		Now, $(2g+k+2m-2c)\circ\tau_n$ is the unique connected non-orientable morphism $n\rightarrow 1$ in \catb with precisely $2g+k+2m-2c+1$ crosscaps. On the other hand, the genus of $\Sigma_{g,k,c}$ increases by precisely $m-c$ on composition with $\tau_m$, and the resulting surface is connected. Therefore $\tau_m\circ\Sigma_{g,k,c}$ is non-orientable with $2(g+m-c)+k+1=2g+k+2m-2c+1$ crosscaps. Note that, by ensuring that both compositions are non-orientable, we need not worry about the direction of boundary maps. Therefore the diagram commutes as required.
	\end{proof}
%
%
\section{Localisations}
%
%
%
		For any category $\mathscr{C}$, we denote its \emph{localisation}, that is the groupoid obtained from $\mathscr{C}$ by formally adjoining inverses of all morphisms,  by $\mathscr{C}[\mathscr{C}^{-1}]$ (see \cite{GZ} for formal construction). Localisation of categories is a generalisation of the notion of group completion of monoids: considering a monoid $M$ as a category with one object then $M[M^{-1}]$ is the group completion $\mathcal G (M)$ of $M$. 

		The classification of invertible TFTs in the next section is based on our calculations here of the localisations of the cobordism categories. As a first step we  prove a general result relating the automorphisms of an object in the localisation of a category $\mathscr C$ to the group completion of the monoid of endomorphisms in $\mathscr C$ of the same object.
		\begin{definition}
			A category $\mathscr{C}$ is \emph{strongly connected} at an object $x$ if for any  object $y$ in $\mathscr{C}$ there exists morphisms $x\rightarrow y$ and $y\rightarrow x$.
		\end{definition}
%
%
		\begin{proposition}
		\label{stronglyconnected}
			If $\mathscr{C}$ is  strongly connected category at $x$, then the canonical map 
			\[
				\mathscr{G}(\operatorname{End}_{\mathscr{C}}(x)) \rightarrow \operatorname{Aut}_{\mathscr{C}[\mathscr{C}^{-1}]}(x)
			\] 
			is surjective.
		\end{proposition}
		\begin{proof}
			Consider a general automorphism  $F$ of $x$ in $\mathscr{C}[\mathscr{C}^{-1}]$, which is a composition of morphisms in $\mathscr{C}$ and their inverses. This is represented by the top line of the diagram below, where $t,u,v,w$ are other objects in $\mathscr{C}$.
			\begin{displaymath}
				\xymatrix{
					\overset{x}{\bullet}\ar[r] & \overset{t}{\bullet}\ar@/_0.5pc/[d]\ar@/^0.5pc/[d] & \overset{u}{\bullet}\ar[l]\ar[r] & \overset{v}{\bullet}\ar@/_0.5pc/[d]\ar@/^0.5pc/[d] & \overset{w}{\bullet}\ar[l]\ar[r] & \overset{x}{\bullet} \\
					& \underset{x}{\bullet} & \underset{x}{\bullet}\ar@/^0.5pc/[u]\ar@/_0.5pc/[u] & \underset{x}{\bullet} & \underset{x}{\bullet}\ar@/^0.5pc/[u]\ar@/_0.5pc/[u]&}
			\end{displaymath}
			We move from left to right along the diagram. Since $\mathscr{C}$ is strongly connected at $x$, when we reach the object $t$ we can find a morphism $t\rightarrow x$ in $\mathscr{C}$, represented by the first downwards arrow. We then map back to $t$ via the inverse of this morphism. Similarly, when we reach $u$ we can find a morphism $x\rightarrow u$, and map down to $x$ via its inverse. Continuing in this manner, we can decompose $F$ as a composition of endomorphisms of $x$ in $\mathscr{C}$ and their inverses, that is to say as elements of $\mathscr{G}(\operatorname{End}_{\mathscr{C}}(x))$.
		\end{proof}
		
	\subsection{Localisations of subcategories of \cat}
	\label{closedsubcats}
%
%
		Before looking at the localisation of the whole category $\cat$, we will first study the localisation of some important subcategories. In particular we will need the following result.
		\begin{theorem}
		\label{prop: localisation of catone}
			$\catone[\catone^{-1}]=\mathbb{Z}$.
		\end{theorem}

		\begin{proof}
			By our identification of $\catone$ in section 2, to prove the theorem is equivalent to proving that the group completion of $\frac{\nn\times\mathbb{Z}_2}{\sim}$ is $\mathbb{Z}$.

			Recall that the group completion of an abelian monoid $M$ is given by the Grothendieck construction $\frac{M\times M}{\approx}$ where the relation is $(x,y)\approx(x',y')$ if and only if there exists $k\in M$ such that $x+y'+k=y+x'+k$ in $M$. Here we denote the relation on $M\times M$ by $\approx$ to distinguish it from the relation ${\sim}$ on $\nn\times\mathbb{Z}_2$. Therefore, the group completion of $\frac{\nn\times\mathbb{Z}_2}{\sim}$ is given by $\frac{(\frac{\nn\times\mathbb{Z}_2}{\sim})\times(\frac{\nn\times\mathbb{Z}_2}{\sim})}{\approx}$ where $((a,b,\epsilon),(c,d,\eta))\approx((a',b',\epsilon'),(c',d',\eta'))$ if and only if there exists $(x,y,\zeta)\in\frac{\nn\times\mathbb{Z}_2}{\sim}$ such that $(a,b,\epsilon)+(c',d',\eta')+(x,y,\zeta)\sim(c,d,\eta)+(a',b',\epsilon')+(x,y,\zeta)$ in $\frac{\nn\times\mathbb{Z}_2}{\sim}$.
			
			We will show that the following two relations hold.
			\begin{IEEEeqnarray*}{LCCLR}
				\text{(i)}\hspace{1cm} & ((a,b,0),(c,d,\eta)) & \approx & ((a,b,1),(c,d,\eta)) & \text{for all } a,b,c,d,\eta\\
				\text{(ii)} & ((a,b,0),(c,d,0))      & \approx & ((0,2a+b,0),(c,d,0)) & \text{for all } a,b,c,d
			\end{IEEEeqnarray*}
			It will follow that $\mathscr{G}(\frac{\nn\times\mathbb{Z}_2}{\sim})$ is isomorphic to $\frac{(\{0\}\times\mathbb{N}\times\{0\})\times(\{0\}\times\mathbb{N}\times\{0\})}{\approx}$. The relation $\approx$ can then be written as $((0,b,0),(0,d,0))\approx((0,b',0),(0,d',0))$ if and only if there exists $(x,y,\zeta)\in\frac{\nn\times\mathbb{Z}_2}{\sim}$ such that $(x,b+d'+y,\zeta)\sim(x,d+b'+y,\zeta)$, which is the case if and only if $b+d'+y=d+b'+y$ in $\mathbb{N}$. It will follow that $\mathscr{G}(\frac{\nn\times\mathbb{Z}_2}{\sim})$ is isomorphic to the group completion $\mathscr{G}(\mathbb{N})=\mathbb{Z}$.

			For the first equivalence (i), take $(x,y,\zeta)=(0,1,0)\in\frac{\nn\times\mathbb{Z}_2}{\sim}$. Then $(a,b,0)+(c,d,\eta)+(0,1,0)=(a+c,b+d+1,\eta)\sim(c+a,d+b+1,\eta+1)=(c,d,\eta)+(a,b,1)+(0,1,0)$, where for the second step we have used the fact that $b+d+1>0$. Equivalence (ii) follows in a similar fashion, again by taking $(x,y,\zeta)$ to be $(0,1,0)$.
		\end{proof}

		\begin{remark}
			Note that equations (i) and (ii) tell us that in $\catone[\catone^{-1}]$ the morphism $C_r^{S^1}$ is equivalent to the identity, and the  torus with two discs removed is equivalent to the Klein bottle with two discs removed.
		\end{remark}
	
		The localisations of $\catoneor$ and $\catonenor$ can be computed via the same methods as were used for \catone. 

		\begin{proposition}
			$\catoneor[(\catoneor)^{-1}]=\mathbb{Z}\times\mathbb{Z}_2$ and $\catonenor[(\catonenor)^{-1}]=\mathbb{Z}$.
		\end{proposition}

		\begin{proof}
			For the first statement recall that we identified $\catoneor$ with the monoid $\mathbb{N}\times\mathbb{Z}_2$. This group completes to $\mathbb{Z}\times\mathbb{Z}_2$, generated by the torus with two discs removed of type $0$, and the cylinder $C_r^{S^1}$. For the second statement, recall that we identified $\catonenor$ with the monoid $\frac{\mathbb{N}\times\mathbb{Z}_2}{\sim}$ where the relation is $(k,0)\sim(k,1)$ for non-zero $k$. Via the same methods as in the proof of Theorem 3.3, one can show that the group completion of this monoid is $\mathbb{Z}$, generated by the projective plane with two discs removed.
		\end{proof}

			Finally, we define the category \cato to be the full subcategory of \cat on one object, the empty $1$-manifold $\emptyset$.

		\begin{proposition}
		\label{localisationofcato}
	        $\cato[\cato^{-1}]=\mathbb{Z}^{\infty}$.
		\end{proposition}

		\begin{proof}
	        Morphisms in \cato are closed surfaces, and are completely determined by the genus $g= 0, 1, \dots$  or number of crosscaps $k=1,2,\dots$  of each of their components. A general morphism in \cato is thus just the union of an element of $\mathbb{N}^{\mathbb{N}}$ and an element of $\mathbb{N}^{\mathbb{N}_{>0}}$, with composition corresponding to component-wise addition in each monoid. We can therefore identify the whole category \cato with the monoid $\mathbb{N}^{\mathbb{N}}\times\mathbb{N}^{\mathbb{N}_{>0}}$. Hence $\cato[\cato^{-1}]=\mathbb{Z}^{\mathbb{N}}\times\mathbb{Z}^{\mathbb{N}_{>0}}=\mathbb{Z}^{\infty}$.
		\end{proof}
	\subsection{Localisations of $\cat$ and $\mathcal S \cap \cat$}
	\label{fundgroupclosed}
%
%

		The localisation of a symmetric monoidal category is again symmetric monoidal, see Appendix A.
		The purpose of this section is to prove  the following theorem.
		\begin{theorem}
		\label{thm: localisation of cat}
			$\lcat \defeq\cat[\cat^{-1}] $ is equivalent to $\mathbb{Z}$ as a symmetric monoidal category.
		\end{theorem}
	
		The proof will show that this equivalence is induced by the Euler characteristic, or more precisely by the functor $\Theta$. It follows in particular that the automorphism group of an object is generated by a surface of Euler characteristic $1$ or $-1$.

		\begin{proof}
			Denote by \lcatn the group of automorphisms of the object $n$ in \lcat. Since \cat is strongly connected for all objects, all \lcatn are isomorphic and the inclusion $\lcatn\hookrightarrow\lcat$ is an equivalence of categories \cite{Q}. The proof proceeds by calculating \lcatone, the automorphism group of the circle. 
			\noindent
			{\it Step 1:}
			\/We first describe the monoid of endomorphisms $\text{End}_{\cat} (S^1)$.
			Disregarding closed components, an element $\Sigma$ in $\text{End}_{\cat}(S^1)$ consists of either a connected surface with two boundary circles (that is, an element of \,\catone) or of two connected surfaces each with one boundary circle (see Figure~\ref{EndS1}). Each component has some arbitrary genus $g\in\mathbb{N}$ or number $k\in\mathbb{N}$ of crosscaps. 
			Connected endomorphisms  as in the left of Figure~\ref{EndS1} also have an associated type which lies in the group $\mathbb{Z}_2$, as discussed in Definition 2.2. Via reflection in a suitable horizontal plane, one sees that the disc with boundary map $i$ is homeomorphic (relative to the boundary) to the disc with boundary map $r$, and so we need not worry about the direction of boundary maps for disconnected endomorphisms as in the right of Figure~\ref{EndS1}.
			
			\begin{figure}[htp]
				\begin{center}
					\includegraphics[scale=0.5]{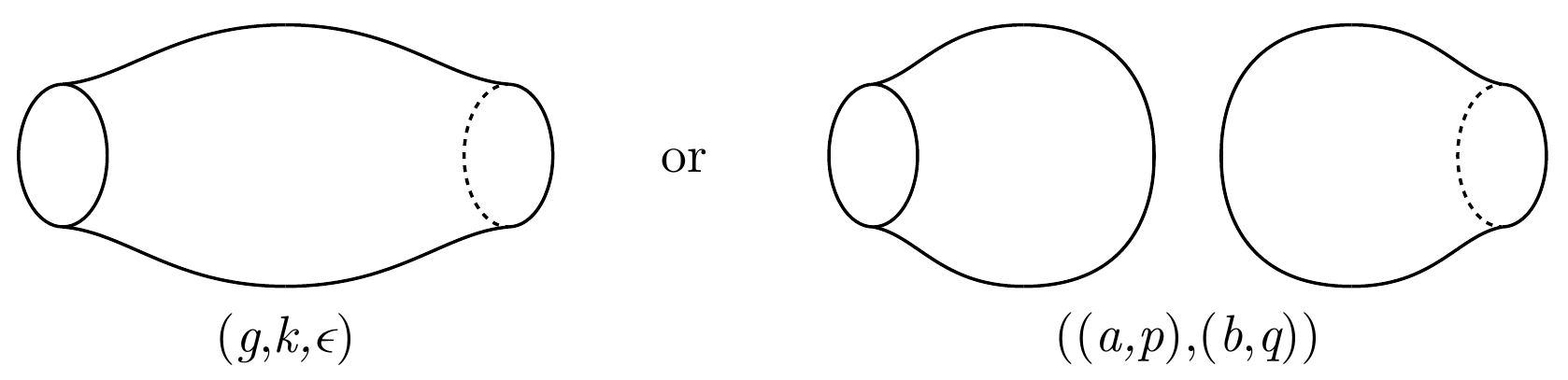}
					\caption{Two kinds of morphisms in $\text{End}_{\cat}(S^1)$.}\label{EndS1}
				\end{center}
			\end{figure}
			The closed components of $\Sigma$ lie in the monoid $\cato=\mathbb{N}^{\mathbb{N}}\times\mathbb{N}^{\mathbb{N}_{>0}}$ from the previous section. It follows that $\text{End}_{\cat}(S^1)$ is isomorphic to 
			\begin{equation*}
						\left(\left(\tfrac{\nn\times\mathbb{Z}_2}{\approx}\right) \amalg \left(\left(\tfrac{\nn}{\sim}\right) \times \left(\tfrac{\nn}{\sim}\right)\right)\right) \times \left(\mathbb{N}^{\mathbb{N}} \times \mathbb{N}^{\mathbb{N}_{>0}}\right)
			\end{equation*}
			where $\approx$ is the relation $(g,k,0)\approx(0,2g+k,1)$ for $k>0$ and $\sim$ is the relation $(g,k)\sim(0,2g+k)$ for $k>0$. Elements of the form $((g,k,\epsilon);n_{00},n_{10},n_{01},\ldots)$ represent morphisms comprising a connected surface as in the left of Figure~\ref{EndS1}, along with $n_{00}$ spheres, $n_{10}$ tori, $n_{01}$ projective planes, $n_{20}$ double tori, $n_{02}$ Klein bottles, and so on.  Elements of the form $((a,p),(b,q);n_{00},n_{10},n_{01},\ldots)$ represent morphisms comprising a disconnected surface as in the right of Figure~\ref{EndS1}, along with $n_{00}$ spheres, $n_{10}$ tori etc. Addition in the monoid is non-commutative and defined by the geometry.
			Since $\catone\subset\operatorname{End}_{\cat}(S^1)$, relations which hold in $\mathscr{G}(\catone)$ must also hold in $\mathscr{G}(\operatorname{End}_{\cat}(S^1))$, and hence also in \lcat. In particular the cylinder $C_r^{S^1}$ is equivalent to the identity $1_{S^1}$ by the proof of Theorem 3.3,  and so the direction of boundary maps becomes irrelevant upon localisation. Since we will be working in \lcat from now on, we will henceforth suppress all mention of boundary maps.
			\noindent
			{\it Step 2:}
			\/We next  eliminate all closed components other than spheres.
			Since \cat is connected, conjugation $c_{\alpha}\colon \lcatone\rightarrow\lcato$ defined by $\beta\mapsto\alpha^{-1}\beta\alpha$ is an isomorphism for any $\alpha\colon 0\rightarrow 1$. In particular it is injective, and we can use this to obtain relations on \lcatone by finding elements with homeomorphic images in \lcato. Let $\alpha$ be the disc as a morphism $0\rightarrow 1$. A representative for the inverse of $\alpha$ in \lcat is the disc $\colon 1\rightarrow 0$ union the inverse of a sphere. The images of elements of \lcatone under $c_{\alpha}$ are:
			\begin{align*}
				\pm((g,k);n_{00},n_{10},n_{01},\ldots) & \mapsto\pm(n_{00}-1,n_{10},n_{01},\ldots,n_{gk}+1,\ldots)\\
				\pm((a,p),(b,q);n_{00},n_{10},n_{01},\ldots) & \mapsto\pm(n_{00}-1,n_{10},n_{01},\ldots,n_{ap}+1,\ldots,n_{bq}+1,\ldots)
			\end{align*}
			where $g,k,a,b,p,q,n_{i0},n_{0j}$ are non-negative and $\pm$ represents an element of $\operatorname{End}_{\cat}(S^1)$ or its inverse respectively.
			Injectivity of $c_{\alpha}$ thus forces the following identifications in \lcatone:

			\vspace{-20pt}
			\begin{IEEEeqnarray*}{RLL}
				\text{(i)} & \pm ((g,k);n_{00},n_{10},\ldots)\\
				& =\pm((g+\sum_iin_{i0},k+\sum_jjn_{0j});\sum_in_{i0}+\sum_jn_{0j},0,0,\ldots )\\
				&=\pm ((0,2g+k+2\sum_iin_{i0}+\sum_jjn_{0j});\sum_in_{i0}+\sum_jn_{0j},0,0,\ldots )\\
				\text{(ii)} & \pm((a,p),(b,q);n_{00},n_{10},\ldots) \\ &=\pm ((a+b+\sum_iin_{i0},p+q+\sum_jjn_{0j});1+\sum_in_{i0}+\sum_jn_{0j},0,0,\ldots )\\
				&=\pm ((0,2a+2b+p+q+2\sum_iin_{i0}+\sum_jjn_{0j});1+\sum_in_{i0}+\sum_jn_{0j},0,0,\ldots ).
			\end{IEEEeqnarray*}
			Here we use `$=$' rather than the usual `$\sim$' to mean equivalent, in order to avoid confusion with the relation $\sim$ defined earlier.
			For the first step in each equivalence we have used the fact that $c_\alpha$ is a group homomorphism.
			For the second step  we have used the fact that the torus  and the  Klein bottle each  with two discs removed are equivalent in $\mathscr{G}(\catone)$ (by the proof of Theorem 3.3) and therefore also in \lcatone. We see that every morphism in \lcatone is equivalent to an element of $\catonenor$ union a collection of spheres, or its inverse, or a composition of such things.
			\noindent
			{\it Step 3:}	
			\/The next step is to eliminate the spheres, meaning that every element of \lcatone is equivalent to an element of $\mathscr{G}(\catonenor)$.
			We show that 
			\[
				((0,-2);0,0,\ldots)=((0,0);1,0,0,\ldots)
			\] 
			in \lcat, in other words the inverse of the Klein bottle with two discs removed is a cylinder union a sphere, or by our earlier identification, the inverse of the torus with two discs removed is a cylinder union a sphere, for which we follow  \cite[Theorem 7]{T}.
			 
			We recall that conjugation $c_{\alpha}\colon \lcatm\rightarrow\lcatn$ via $\beta\rightarrow\alpha^{-1}\beta\alpha$ is a group isomorphism for any $\alpha\colon n\rightarrow m$. Let $\alpha$ be the union of a cylinder and a disc as a morphism $1\rightarrow 2$. We consider the following two representatives for the inverse of $\alpha$ in \lcat: let $\beta_1\colon 2\rightarrow 1$ be the pair of pants surface, and let $\beta_2\colon 2\rightarrow 1$ be the union of a disc, a cylinder and the inverse of a sphere. In order for conjugation to be well defined, we must have $\beta_1\gamma\alpha=\beta_2\gamma\alpha$ in \lcat for any morphism $\gamma\colon 2\rightarrow 2$. Taking $\gamma$ to be the union of the pair of pants surface and a disc, we see that the twice punctured torus union a sphere is the identity (see Figure~\ref{Conjugation}).
			By composing each side of the above equation with $((0,1);0,0,\ldots)$, the  projective plane with two discs removed, it follows that in \lcat $((0,-1);0,0,\ldots) = ((0,1);1,0,0,\ldots)$ and more generally
			\[
				((0,-k);0,0,\ldots)=((0,k);k,0,0,\ldots),
			\]
			and
			\begin{align*}
				((0,k);n_{00},0,0,\ldots) & =((0,k-n_{00});0,0,\ldots)\circ((0,n_{00});n_{00},0,0,\ldots)\\
				& = ((0,k-n_{00});0,0,\ldots)\circ((0,-n_{00});0,0,\ldots)\\
				& = ((0,k-2n_{00});0,0,\ldots). 
			\end{align*}

			\begin{figure}[htp]
				\begin{center}
					\includegraphics[scale=0.38]{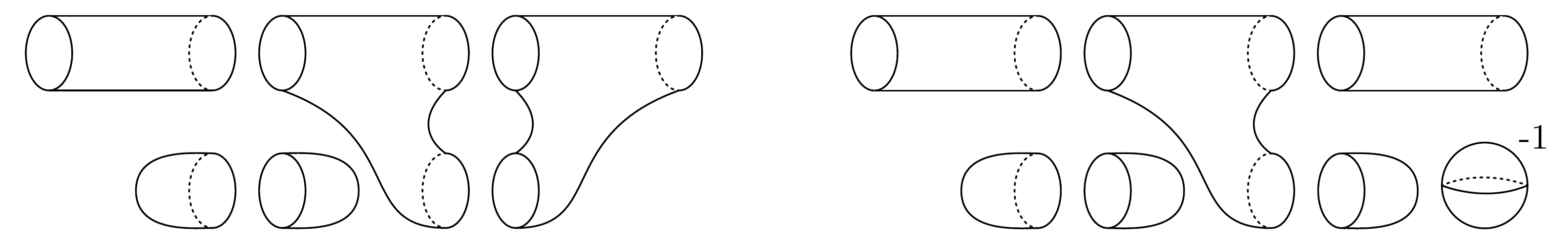}
					\caption{$c_{\alpha}(\gamma)=\beta_1\gamma\alpha$ and  $c_{\alpha}(\gamma)=\beta_2\gamma\alpha$.}\label{Conjugation}
				\end{center}
			\end{figure}
			Applying these identities to the equivalences (i) and (ii) we obtain
			\vspace{6pt}
			\begin{IEEEeqnarray*}{RL}
				\text{(i)$'$} & \pm((g,k);n_{00},n_{10},\ldots)\\
				&=\pm ((0,2g+k+2\sum_i(i-1)n_{i0}+\sum_j(j-2)n_{0j});0,0,\ldots )\\
				\text{(ii)$'$} & \pm((a,p),(b,q);n_{00},n_{10},\ldots)\\
				&=\pm ((0,2(a+b-1)+(p+q)+2\sum_i(i-1)n_{i0}+ \sum_j(j-2)n_{0j});0,0,\ldots)
			\end{IEEEeqnarray*}
			and so every element of \lcatone is of the desired form. It remains to show that there are no further relations.
	
			\noindent
			{\it Step 4:}
			\/The functor $\Theta\colon \cat\rightarrow\mathbb{Z}$ is morphism inverting and hence factors through \lcat by the universal property of localisations. Let $\bar {\Theta}$  denote the unique functor $\lcat\rightarrow\mathbb{Z}$ corresponding to $\Theta$, so that the following diagram commutes, where $\psi\colon \cat\rightarrow\lcat$ is the canonical projection.
			\begin{displaymath}
				\xymatrix{
					\cat\ar[rr]^{\psi}\ar[dr]_{\Theta} & & \lcat\ar[dl]^{\bar{\Theta}}\\
					& \mathbb{Z} & 
				}
			\end{displaymath}
			Since $\bar{\Theta}((0,\pm k);0,0,\ldots)=\mp\chi((0,k);0,0,\ldots)=\pm k$  we see that $\bar {\Theta}|_{\lcatone}\colon \lcatone\rightarrow\mathbb{Z}$ is an isomorphism. Hence all further relations on \lcatone are trivial, and therefore $\lcat\sim\lcatone=\mathbb{Z}$.

			Finally we note that ${\Theta} $ is a map of  symmetric monoidal categories by Proposition 2.4, and so is $\bar{\Theta}$. Indeed, they are strict monoidal functors of based, strict symmetric monoidal categories. Let $\Phi \colon \mathbb {Z} \rightarrow \cat$ be the functor that assigns to the integer $k$	the endomorphism of the empty 1-manifold consisting of $k$ copies of the projective plane. Both $\Phi$ and $\bar \Phi \defeq \psi \circ \Phi$ are based symmetric monoidal, and the pair $(\bar \Theta, \bar \Phi)$ gives rise to a based symmetric monoidal equivalence in the sense of Definition A.7.  
			\qedhere
		\end{proof}
		The subcategory  of orientable surfaces in $\cat$ was considered in \cite {T}.

		\begin{theorem}
			The localisation of $\mathcal S \cap \cat$ is equivalent to $\Bbb Z$ as a symmetric monoidal category.
			\qed
		\end{theorem} 

		Indeed, by restricting to only orientable surfaces  the  proof of Theorem 3.6 will reprove this result and identify the torus with two discs removed as a  generator of the automorphism group of the circle. In \lcat this is equivalent to the  Klein bottle with two discs removed. Thus the inclusion $\mathcal {S} \cap \cat \rightarrow\cat$ induces on localisations the multiplication by $2$ map.  It is somewhat surprising that the morphism $C_r^{S^1}$ plays no part in the localisation of \cat;  one might have expected a $\mathbb{Z}_2$ factor to appear.


	\subsection{Localisations of \catoo and $\mathcal S \cap \catoo$  }

		Let \catop be the full subcategory of \catoc on objects that are disjoint unions of intervals. We defined the open category \catoo to be the subcategory of \catop containing all its objects, and those morphisms which do not have any closed components. We can view the category \cato  as a subcategory of \catop; it is the subcategory of closed endomorphisms of the empty $1$-manifold $0$. Note that there is  a Cartesian product decomposition
		\[
			\catoo\times\cato\simeq \catop
		\]  
		of symmetric monoidal categories and that localisation commutes with Cartesian product. We thus turn our attention to $\catoo$. 

 		\begin{theorem}
		\label{thm: localisation of catoo}
			$\lcatoo \defeq \catoo[\catoo^{-1}]$ is equivalent to $\mathbb{Z}$ as a symmetric monoidal category.
        \end{theorem}
 		
		\begin{proof}
	        For each object $n\in\catoo$ define $\lcatoon\defeq\operatorname{Aut}_{\lcatoo}(n)$. The category \catoo is strongly connected for every object. We proceed as in the proof  of Theorem 3.6. As the arguments are similar, we give only the essential steps.

			\noindent
			{\it Step 1:} \/We will first describe the monoid $\operatorname{End}_{\catoo}(I)$.
	        Disregarding components with entirely free boundary, elements of this monoid take one of three forms depicted in Figure~\ref{EndI}. For each cobordism in the picture, the thickened lines represent the source and target boundaries of the morphism. The thin lines represent free boundary. The morphism on the left is the disc as an endomorphism of the interval $I$, and the centre morphism is two discs. The morphism on the right is the cylinder as an endomorphism of $I$, and we often refer to this surface as the \emph{whistle}. Each component has some genus $g$, number of crosscaps $k$, and number of windows $w$.  Components of a morphism which have entirely free boundary are discs, also with genus, crosscaps and windows. We will refer to the standard disc $D^2$ with entirely free boundary as the \emph{free} disc.
			Next we show that, if we are only interested in the localisation, we do not need to mention boundary maps. Indeed,  
	conjugation $c_{\alpha}\colon\lcatooone\rightarrow\lcatooo$ is an isomorphism for any $\alpha\colon 0\rightarrow 1$. Let $\alpha$ be the disc as a morphism $0\rightarrow 1$, and note that an inverse for $\alpha$ is the disc $\colon 1\rightarrow 0$ union the inverse of a free disc. Under $c_{\alpha}$, the disc $C_r^I$ and the identity $1_I$ both map to $\emptyset$. It follows that $C_r^I$ and $1_I$ are equivalent in $\lcatoo$.
			\begin{figure}[htp]
				\begin{center}
					\includegraphics[scale=0.7]{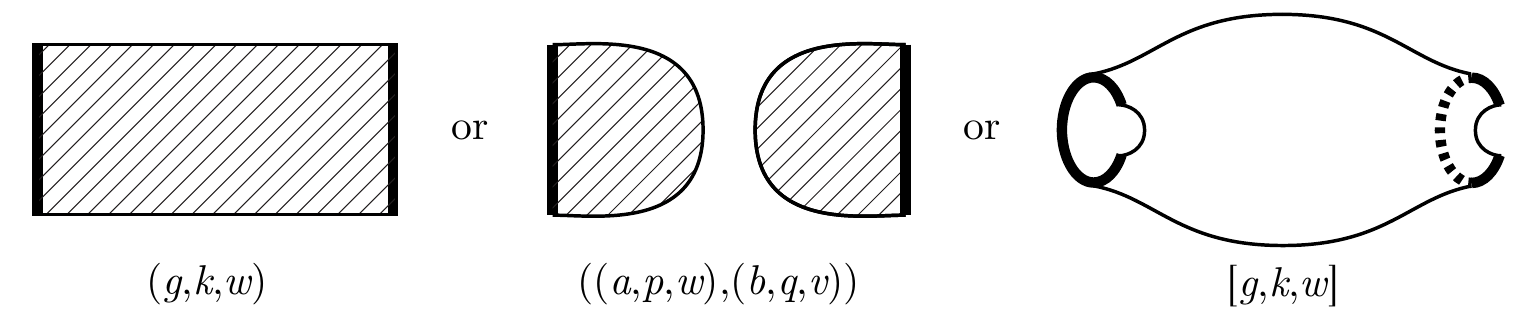}
					\caption{Three kinds of morphisms in $\text{End}_{\catoo}(I)$.}\label{EndI}
				\end{center}
			\end{figure}

			From the above discussion (and ignoring type) we can write the monoid $\operatorname{End}_{\catoo}(I)$ as
			\[
				\left(\left(\tfrac{\nn \times \mathbb{N}}{\sim}\right) \amalg \left(\left(\tfrac{\nn \times \mathbb{N}}{\sim}\right) \times \left(\tfrac{\nn \times \mathbb{N}}{\sim}\right)\right) \amalg \left(\tfrac{\nn \times \mathbb{N}}{\sim}\right)\right) \times \left(\mathbb{N}^{\mathbb{N} \times \mathbb{N}} \times \mathbb{N}^{\mathbb{N}_{>0} \times \mathbb{N}}\right)
			\]
			where the relation on each component is $(g,k,w)\sim(0,2g+k,w)$ for non-zero $k$. Addition in the monoid is non-commutative and defined by the geometry. Elements are written in one of the following ways:

			\begin{itemize}
				\item $((g,k,w);n_{000},n_{100},n_{010},\ldots,n_{001},n_{101},n_{011},\ldots,\ldots)$ represents a connected morphism as in the left of Figure~\ref{EndI} along with $n_{000}$  free discs, $n_{100}$ punctured tori, $n_{010}$ M\"obius bands,\ldots,$n_{001}$ annuli, $n_{101}$ twice punctured tori, $n_{011}$  twice punctured projective planes, and so on ($n_{ijt}$ is the number of discs with $i$ handles, $j$ crosscaps, and $t$ windows).
				\item $((a,p,w),(b,q,v);n_{000},n_{100},n_{010},\ldots,n_{001},n_{101},n_{011},\ldots,\ldots)$ represents a disconnected morphism as in the centre of Figure~\ref{EndI} along with $n_{000}$ free discs, $n_{100}$ punctured tori, and so on.
				\item $([g,k,w];n_{000},n_{100},n_{010},\ldots,n_{001},n_{101},n_{011},\ldots,\ldots)$ represents a morphism as in the right of Figure~\ref{EndI} along with $n_{000}$ free discs etc, where we use square brackets to distinguish the whistle from the disc in the marked component.
			\end{itemize}

			\noindent
			{\it Step 2:}
			\/We next  show that we may concentrate on one of the three kinds and eliminate all components with only free boundary other than the disc. This is achieved by  considering the images under the conjugation map $c_{\alpha}\colon \lcatooone\rightarrow\lcatooo$ as defined in Step 1. Let $\Sigma (g,k, w )$ denote a surface of (orientable) genus $g$ with $k$ crosscaps and $w +1 $ windows, i.e.\ the surface represented by $n_{g,k,w}$ in the above notation. We note that $c_\alpha (((g, k, w); 0, \dots))$ is $ \Sigma (g,k, w)$ minus a free disc, $c_\alpha (([g,k,w]; 0 \dots))$ is  $ \Sigma (g,k,w+1) $ minus a free disc, and $c_\alpha (((a,p,w), (b,q,v); 0, \dots))$ is $ \Sigma (a,p ,w) \cup \Sigma (b,q,v)$ minus a free disc.Calculations as in the proof of Theorem 3.6 show that in the localisation all morphisms can be represented by morphisms of the form $((g,k,w); n_{000}, 0, 0, \dots)$ and their inverses, and hence, as $((g,k,w); 0, \dots)) = ((0, 2g+k,w); 0, \dots)$ for $k>0$, those of the form $((0,k,w); n_{000}, 0, 0, \dots)$ and their inverses.   

			\noindent
			{\it Step 3:}
			\/We will deduce two more relations which will show that $\lcatoo_1 $ is a quotient of $\mathbb{Z}$.

			First we consider the  conjugation $c_{\alpha}\colon \lcatootwo\rightarrow\lcatooone$ where $\alpha$ is the union of two discs as a morphism $1\rightarrow 2$ as in Figure~\ref{Conjugation1}. The following are both inverses for $\alpha$. Let $\beta_1$ be the disc as a morphism $2\rightarrow 1$, and let $\beta_2\colon 2\rightarrow 1$ be the union of two discs and an inverse free disc (see Figure~\ref{Conjugation1}). Let $\gamma\colon 2\rightarrow 2$ be two discs as in the figure. Since we must have $\beta_1\gamma\alpha=\beta_2\gamma\alpha$ in \lcatoo, we see that the annulus union a free disc is equivalent to the identity, both as morphisms $1\rightarrow 1$.
			Thus in $\lcatoo$, 
			\[
				((0,0,-1);0,0,\ldots)=((0,0,0);1,0,0,\ldots).
			\] 

			Next we consider conjugation $c_{\alpha}\colon \lcatootwo\rightarrow\lcatooone$ just as above only that we choose a different $\beta_1$. This time $\beta_1$ is the disc as a morphism $2 \rightarrow 1$ where one of the incoming boundary maps is twisted. Note that this still defines an inverse of $\alpha$. This time the  surface $c_\alpha(\gamma)$ is a M\"obius band union a free disc. 
			Thus in $\lcatoo$, 
			\[
				((0,-1,0);0,0,\ldots)=((0,0,0);1,0,0,\ldots).
			\] 
			\begin{figure}[htp]
				\begin{center}
					\includegraphics[scale=0.45]{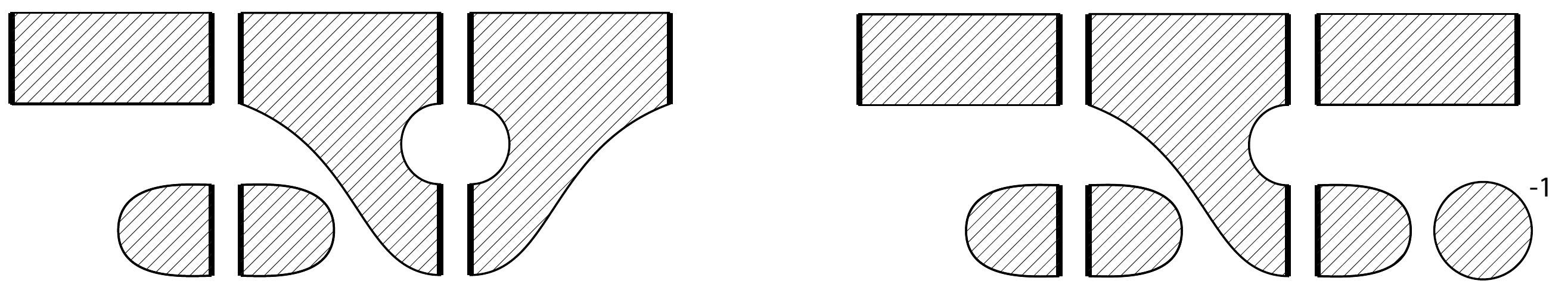}
					\caption{$c_{\alpha}(\gamma)=\beta_1\gamma\alpha$ and  $c_{\alpha}(\gamma)=\beta_2\gamma\alpha$.}\label{Conjugation1}
				\end{center}
			\end{figure}
		
			\noindent
			{\it Step 4:}
			\/The functor $\bar \Theta\colon \lcatoo \rightarrow \mathbb{Z}$ is surjective, mapping the M\"obius band and the annulus to $1\in \mathbb Z$ when considered as morphisms $1 \rightarrow 1$. Hence, in particular, 
			$\lcatoo_1$ is isomorphic to $\mathbb{Z}$.
		\end{proof}
		
		Consider now the subcategory of orientable surfaces  in $\catoo$.

		\begin{theorem}
			The localisation of $\mathcal S \cap \catoo$ is equivalent to $\Bbb Z$ as a symmetric monoidal category.
		\end{theorem}
		
		\begin{proof}
			We adopt the proof of Theorem 3.8. The description of the endomorphisms of the interval simplifies in $\mathcal S \cap \catoo$	and we denote  the three kinds of surfaces by $((g, w); n_{000}, \dots) $, $ ((a,w), (b,v); n_{000} \dots)$ and $([g,w]; n_{000} , \dots )$. By the same argument as in Step 2 and the first part of Step 3 above we may reduce our attention to surfaces of the form $ ((g,0); 0, \dots))$ and $((0,w); 0, \dots)$ and their inverses. We will show that in the localisation of $\mathcal S \cap \catoo$
			\[
				((1, 0); 0, \dots) =  ((0,-2); 0, \dots ).
			\]   
			As Step 4 is still valid, this will prove the theorem.

			We work  in the localisation and note that $((0,1); 1, 0, \dots)$ is the identity by the analogue of the first part of Step 3. Composing this with a cylinder considered as a morphism from the interval to itself shows that	$([0,0]; 0,0 \dots) = ([0,w]; w, 0, \dots)$ for any $w$. From this identity it now follows in particular that$ \beta_1 \alpha=  ([0, 2]; 2,0 \dots) =  ([0,1]; 1,0 \dots )= \beta _2 \alpha$ for the morphisms $\alpha, \beta _1 $ and $\beta_2$ as depicted in Figure 3.5. Hence, the two composed surfaces $\beta_1 \gamma \alpha $ and $\beta_2 \gamma \alpha$ in Figure 3.5 are identified. This gives
			\[
				([1,3]; 3, 0, \dots) = ([0, 3]; 1, 0 ,\dots).
			\]
			Finally we note that $([1,3]; 3, \dots 0) = ([0, 3]; 1, 0 , \dots ) ( (1,0); 2, 0, \dots)$ which implies  $((1,0); 2,0 \dots)$ and hence $((1,2); 0, \dots )$ is the identity.
			\qedhere
		\end{proof}

		\begin{figure}[htp]
			\begin{center}
				\includegraphics[scale=0.56]{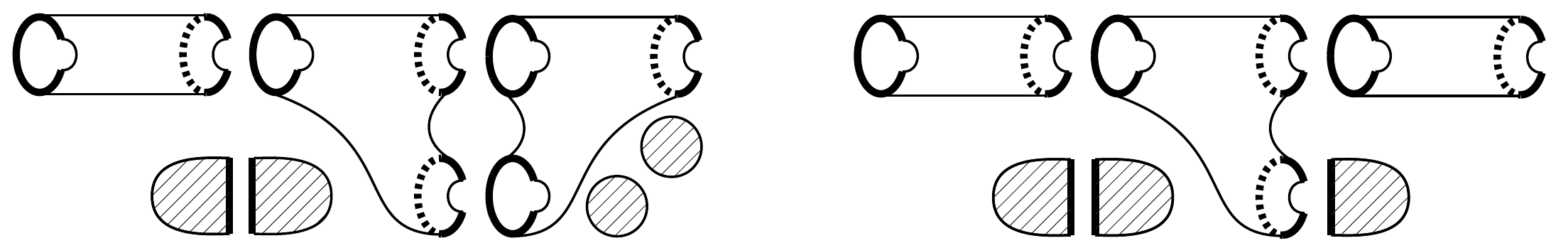}
				\caption{$\beta_1\gamma\alpha$ and  $\beta_2\gamma\alpha$.}\label{Conjugation4}
			\end{center}
		\end{figure}

	\subsection{Localisations of $\catc$, $\catoc$ and $\mathcal S$}

Let \catc be the full subcategory of \catoc on those objects which are closed. Note that \,\catc contains and is closely related to \,\cat, the difference being that in \,\catc we allow morphisms to have windows. In this section we compute the localisation of $\catc$. The results obtained, along with those from previous sections, will enable us to easily compute the localisation of $\catoc$.

		\begin{theorem}
		\label{thm: localisation of catc}
			$\lcatc\defeq\catc[\catc^{-1}]$ is equivalent to $\mathbb{Z}\times\mathbb{Z}$ as a symmetric monoidal category.
		\end{theorem}
		
 		\begin{proof}
			$\catc$ is strongly connected at the object $S^1$, so we proceed by looking at $\mathscr{G}(\operatorname{End}_{\catc}(S^1))$. Note that morphisms in \catc are those of \cat which may have also windows. It is a simple exercise to go through the computation of $\mathscr{G}(\operatorname{End}_{\cat}(S^1))$ in the proof of Theorem 3.3  and check that every morphism in $\mathscr{G} (\operatorname{End}_{\catc}(S^1))$ can be represented by a connected surface $\Sigma_{k,w}$ with $k$ crosscaps and $w$ windows or its inverse or composition of such.

		To see that in $\lcatc$ the variables $k$ and $w$ are independent, define a functor
		\[
			\omega \colon \catc \longrightarrow \mathbb Z
		\]
		which sends any surface to the total number of its windows. Note that $\omega$ is functorial as gluing along circles does not introduce additional windows. Indeed, $\omega$ is a functor of symmetric monoidal categories. The product functor
		\[
			(\Theta - \omega, \omega)\colon \catc \rightarrow \mathbb Z \times \mathbb Z
		\]
		assigns to the representative $\Sigma _{k,w}$ the pair $(k,w)$. It factors through $\lcatc$ where it defines thus an equivalence of symmetric monoidal categories.
		\qedhere
		\end{proof}
		
		We now turn our attention to the whole category \catoc.
		\begin{theorem}
		 \label{thm: localisation of catoc}
                $\lcatoc \defeq \catoc[\catoc ^{-1}]$ is equivalent to $\mathbb Z$ as a symmetric monoidal category. 
        \end{theorem}

		\begin{proof}
			We proceed by looking at $\mathscr{G}(\operatorname{End}_{\catoc}(S^1))$. Endomorphisms of the circle in \catoc take the same form as those in \catc. Thus we deduce that  $\lcatocone=\mathbb{Z}\times\mathbb{Z}$ modulo relations. Here $(s,t)$ denotes the object consisting of $s$ circles and $t$ intervals. In particular, a morphism in \lcatocone  can be represented by a connected surface $\Sigma _{k}$ with $k$ crosscaps, a cylinder $C_w$ with $w$ windows,  their inverses or composition of such things. 

			As \catoo is a subcategory of \catoc, the relations in \lcatoo must also hold in \lcatoc. In \lcatoc the inverse of the cylinder considered as a morphism from the interval to the circle is given by the cylinder considered as a morphism from the circle to the interval union a free disc. Conjugation by this element thus defines an isomorphism $\lcatocone \rightarrow \lcatoc _{(0,1)}$.
		
			In particular, using Theorem 3.8, $\Sigma _{k} = C_k \in \mathbb Z \simeq \lcatoo_1$ and hence they are equal in $ \lcatoc _{(0,1)}$. Thus $\lcatoc_{(1,0)} \simeq \lcatoc_{(0,1)} \simeq \mathbb Z$.
		\end{proof} 

		The analogue of this  theorem in the orientable case was considered in \cite {D}. We give an alternative proof here.

		\begin{theorem}
			The localisation of $\mathcal S$ is equivalent to $\mathbb Z$ as a symmetric monoidal category.
		\end {theorem}

		\begin{proof}
			The proof of Theorem 3.10 can be adapted to prove that the localisation of $\mathcal S \cap \bar \cat$ is equivalent to $\mathbb Z \times \mathbb Z$. In the  endomorphisms of the circle the two generators of $\mathbb Z \times \mathbb Z$ can be represented by the connected surfaces $\Sigma_1$ of genus one and by the  cylinder $C_1$ with one window. We can map the endomorphisms of the circle to the endomorphisms of the interval in $\mathcal S \cap \catoo$ by precomposing and postcomposing with cylinders considered as morphisms from the interval to the circle and vice versa. By Theorem 3.9, $\Theta$ induces an injection on the automorphism group of any object in the localisation of $\mathcal S \cap \catoo$. But $\Theta ( \Sigma _1) = 2 = \Theta (C_2)$ and hence, in the localisation, $\Sigma_1 = C_2$ and $C_1$  generates the automorphism group of the circle.
		\end{proof}
		
\section {Invertible TFTs}

	The computation of the localisations of the cobordism categories allows us to completely determine the invertible field theories defined on the subcategories of $\catoc$. We will focus on the case $\cat$ for more detailed descriptions.   
%

	\subsection{Classification of functors $\cat\rightarrow\mathbb{Z}$}
	\label{classificationtozclosed}
		We have seen that $\Theta$ plays an important role in analysing the cobordism categories. Here we describe all functors $F\colon\cat \rightarrow \mathbb Z$. Note first that $F$ is morphism inverting and hence factors through the localisation \lcat. It is determined by its image on:
		\begin{itemize}
			\item a generator of $\lcat_0=\mathbb{Z}$ 
			\item a set of connecting morphisms $p_k\colon 0\rightarrow k$, one for each object $k \geq 1$ in \cat.
		\end{itemize}
		 We take  $p_k$  to be the union of $k$ discs each as a cobordism from the empty set to $S^1$.
		By Theorem 3.6, the projective plane $P^2$ is  a generator in $\lcat _0$. Thus $F$ is determined by the  sequence of integers $\{ b_0, b_1, \dots \}$ where $F(P^2) = b_0$ and $F(p_k) = b_k$. Conversely, given such a sequence of integers, we can define a functor. For a morphism $\Sigma\colon k\rightarrow k'$ in \cat, put
		\[
			F(\Sigma) \defeq
  			\begin{cases}
				b_{k'}-b_k-b_0\Theta(\Sigma) & \mbox{if } k,k'\geq 1 \\
				b_{k'}-b_0\Theta(\Sigma) & \mbox{if } k=0,\,k'\geq 1\\
				-b_k-b_0\Theta(\Sigma) & \mbox{if }k\geq 1,\,k'=0\\
				-b_0\Theta(\Sigma) & \mbox{if } k=k'=0.
			\end{cases}
		\]

		If we set $a_0\defeq b_0$ and $a_k\defeq\frac{b_k}{k}$ for $k\geq 1$, then $F$ can also be expressed  by the formula
		\[
			F(\Sigma) = a_mm-a_nn-a_0\Theta(\Sigma).
		\]
		$F$ is symmetric strict monoidal if and only if all the $a_i$ are equal for $i\geq 1$. To see this, note that  it is monoidal if and only if $F(p_{k+k'})=F(p_k)+F(p_{k'})$ for all $k$, $k'$. This is the case if and only if $F(p_k)=k F(p_1)$, in other words if and only if $ka_k=ka_1$ for all $k$. We thus have proved the following result.

	 	\begin{proposition}
			Any functor $F\colon \cat\rightarrow\mathbb{Z}$ is determined by where it sends the projective plane  $P^2$ and the connecting morphisms $p_k$. Conversely, for any set of integers $\{b_0,b_1,b_2,\ldots\}$ there exists a functor sending the projective plane to $b_0$, and sending $p_k$ to $b_k$ for all $k\geq 1$. It is symmetric strict monoidal if and only if $b_k = k \, b_1$ for all $k \geq1$.
		\qed
		\end{proposition}

		Note that Proposition 4.1 depends on our particular (skeletal) choice for the cobordism category \cat. This is somewhat  unsatisfactory. We rectify the situation by considering the functor category $[\cat , \mathbb Z]$ up to equivalence: the equivalence class of the functor category does not change when replacing the source or target category by an equivalent category. Hence by Theorem 3.6 we have the following equivalences of categories:
		\[
			[\cat , \mathbb {Z} ] = [\lcat, \mathbb {Z}] \simeq [ \mathbb {Z}, \mathbb {Z} ] = \mathbb {Z}.
		\]
		The final $\mathbb {Z}$ is the discrete category with objects $\mathbb {Z}$ and only identity morphisms. Multiplication of integers corresponds to composition of endo functors. Theorem 3.6 (in conjunction with Proposition A.8) also allows a similar computation for the category of based symmetric monoidal functors, and again we find
		\[
			\operatorname {SymmMon} [ \cat , \mathbb {Z}]_* \simeq \mathbb {Z}.
		\]
Completely analogous arguments give us the same result for \catoo, \catoc and $\mathcal S$ and 
		\[ 
			\operatorname {SymmMon} [ \catc , \mathbb {Z}]_* \simeq \mathbb Z \times \mathbb Z.
		\]
%
%

	\subsection{Classification of invertible TFTs}
	\label{classificationclosed}
		Following Atiyah \cite {A}, topological field theories on a cobordism category are the symmetric monoidal functors to the category of vector spaces with monoidal product given by the tensor product. With reference to appendix A for the definiton of the category of pointed symmetric monoidal functors, we make the following definitions.
	
		\begin{definition}
			Let $\mathscr C$ be a subcategory of $\catoc$. The category of topological field theories on $\mathscr C$ is defined as 
			\[
				\mathscr C -\text{TFT} \defeq \operatorname{SymmMon} [ \mathscr C, \text{Vect}_{\mathbb{C}}]_*.
			\]
			Similarly the category of invertible topological field theories on $\mathscr C$ is defined as
			\[
				\mathscr C-\text{TFT}^\times \defeq \operatorname{SymmMon}[\mathscr C, Pic ( \text{Vect}_{\mathbb{C}})]_*.
			\]
			Recall from Appendix A.3 that the Picard category  $Pic (\mathscr D)$ associated to a symmetric monoidal category $\mathscr D$ has objects that are invertible with repsect to the monoidal product and all invertible morphisms between them.
		\end{definition}

		\begin{theorem}
		\label{thm: main closed}
			For $\mathscr C = \cat, \catoo, \catoc, \mathcal S \cap \cat, \mathcal S \cap \catoo, \mathcal S $ the category $\mathscr C -\text{TFT} ^\times$ is equivalent to the discrete category $\mathbb{C}^\times$ of non-zero complex numbers. For $\mathscr C = \catc$ it is equivalent to $\mathbb{C}^\times \times \mathbb{C}^\times$.
		\end{theorem}

		\begin{proof}
			We consider only the case $\mathscr C = \cat$. The other cases are similar. 
We have the following equivalences of categories:			
			\begin{align*}
				\cat -\text{TFT} ^\times 
				&= \operatorname{SymmMon}[\cat, Pic ( \text{Vect}_{\mathbb{C}})]_* \\
				& = \operatorname{SymmMon} [\lcat, Pic ( \text{Vect}_{\mathbb{C}})]_* \\
				& \simeq \operatorname{SymmMon} [\mathbb {Z}, \mathbb {C} ^\times]  \\
				& = [\mathbb {Z}, \mathbb {C} ^\times] = \mathbb {C} ^\times.
	 		\end{align*}
			The first equality holds by definition. The second equality holds as any invertible functor factors uniquely through the localisation of the source category. By Theorem 3.6, $\lcat \simeq \mathbb{Z}$ and by the example following definition A.12, $ Pic ( \text{Vect}_{\mathbb{C}}) \simeq {\mathbb {C}} ^\times$. The third equivalence therefore follows by an application of Proposition A.8. Functors of abelian groups are symmetric monoidal and are group homomorphisms. Any group homomorphism from $\mathbb Z$ is determined by its image on 1. The latter identity is thus to be interpreted as an identity of sets, or equivalently, of categories with only identity morphisms. \qedhere
		\end{proof}

		In a more hands on approach, we now describe explicitly the category of symmetric monoidal functors from $\cat$ to $\mathbb {C}^\times = Pic (\text{Vect} _{\mathbb{C}}^s)$, the Picard category of the skeleton of the category of complex vector spaces. 

		\begin{theorem}
			Every functor $\cat\rightarrow\mathbb{C}^\times$ is of the form $F^{\boldsymbol\mu}$ for some unique sequence ${\boldsymbol\mu}=(\mu_i)_{i\in\mathbb{N}}$  of non-zero complex numbers with $\mu_0=F^{\boldsymbol\mu}(P^2)$ and $\mu_k^k=F^{\boldsymbol\mu}(p_k)$ for $k\geq 1$. $F^{\boldsymbol\mu}$ sends all objects of \cat to the only object of $\mathbb{C}^\times$, and sends a morphism $\Sigma\colon n\rightarrow m$ in \cat to
			\[
				F^{\boldsymbol\mu}(\Sigma)\defeq\mu_m^m\mu_n^{-n}\mu_0^{-\Theta(\Sigma)}.
			\]
			Furthermore, any such functor can be made into a symmetric monoidal functor in a unique way, and there is a unique natural isomorphism between any two such functors if and only if they take the same value on $P^2$. 
		\end{theorem}

		\begin{proof}
			The same arguments as were used to classify functors $\cat\rightarrow\mathbb{Z}$ prove that the functors are precisely of the form as claimed. We now show that any $F^{\boldsymbol\mu}$ can be given a unique structure of a symmetric monoidal functor $(F^{\boldsymbol\mu},F^{\boldsymbol\mu}_2)$.
		
			By considering the lower two diagrams of Definition~\ref{def: based symmetric monoidal functor}, and as \cat and $\mathbb{C}^\times$ are both strict monoidal categories, we see that the isomorphisms $F^{\boldsymbol\mu}_2\in\mathbb{C}^\times$ must satisfy
			\begin{equation}
			\label{eq: iso is one}
				F^{\boldsymbol\mu}_2(n,0)=F^{\boldsymbol\mu}_2(0,n)=1 \mbox{\ for all $n\in\mathbb{N}$.}
			\end{equation}
			The $F^{\boldsymbol\mu}_2$ must also be natural, and using the above equation we see that
			\[
				F^{\boldsymbol\mu}_2(n,n')=\mu_n^{-n}\mu_{n'}^{-n'}\mu_{n+n'}^{n+n'}.
			\]
			One checks that this also satisfies the symmetry axiom (the upper right hand diagram of Definition~\ref{def: based symmetric monoidal functor}). Therefore every functor $\cat\rightarrow\mathbb{C}^\times$ can be uniquely given the structure of a symmetric monoidal functor.\footnote
				{The functor $(F^{\boldsymbol\mu},F^{\boldsymbol\mu}_2)$ is symmetric strict monoidal if and only if all the $\mu_i$ are equal for $i\geq 1$.}

			We now describe the morphisms in $\operatorname{SymmMon}[\cat,\mathbb{C}^\times]_*$. Suppose we have a based natural transformation $\tau\colon F^{\boldsymbol\mu}\rightarrow F^{\boldsymbol\mu'}$. Then, by considering the projective plane $P^2$ as a morphism $0\rightarrow 0$, naturality of $\tau$ and the fact that $\tau_0=1$ give $\mu_0=\mu_0'$. So suppose this is the case.
			Then we deduce
			\[
				\tau_n=(\mu_n'\mu_n^{-1})^n.
			\]
			Finally, one checks that $\tau$ defines a based monoidal natural transformation $(F^{\boldsymbol\mu},F^{\boldsymbol\mu}_2)\rightarrow(F^{\boldsymbol\mu'},F^{\boldsymbol\mu'}_2)$; in other words the diagram of Definition \ref{def: based monoidal natural transformation} automatically commutes. Hence there is precisely one morphism between any two objects in $\operatorname{SymmMon}[\cat,{\mathbb{C}}^\times]_*$ which are indexed by the same $\mu_0$, and this morphism is an isomorphism. In particular there are no non-trivial automorphisms. 
			\qedhere		
		\end{proof}

		It is known that any  $\cat$-TFT  corresponds uniquely to a commutative Frobenius algebra $A$ with the following additional structure \cite{TT}.
		\begin{itemize}
		\item An involutive automorphism $\mathbf{x}\mapsto \mathbf{x}^*$ which preserves the pairing on $A$, that is $(\mathbf{x}^*)^*=\mathbf{x}$, $(\mathbf{x}\mathbf{y})^*=\mathbf{x}^*\mathbf{y}^*$ and $\langle \mathbf{x}^*,\mathbf{y}^*\rangle=\langle \mathbf{x},\mathbf{y}\rangle$.
		\item An element $U\in A$ satisfying:
			\begin{enumerate}[(i)]
				\item $(\mathbf{a}U)^*=\mathbf{a}U$ for all $\mathbf{a}\in A$
				\item $U^2=\sum\alpha_{ij}a_ia_j^*$ where $\{a_i\}$ is a basis for $A$ and the copairing $\mathbb{C}\rightarrow A\bigotimes A$ is given by $1\mapsto\sum_{ij}\alpha_{ij}a_i\otimes a_j$.
			\end{enumerate}
		\end{itemize}
		In terms of cobordisms, the multiplication on the vector space $F(S^1)=A$ corresponds to the pair of pants surface as a morphism $S^1\amalg S^1\rightarrow S^1$, the unit is given by the image of $1$ under the linear map $F(p_1)\colon \mathbb{C}\rightarrow A$, the pairing corresponds to the composition of the pair of pants with the disc as a morphism $S^1\rightarrow\emptyset$, the involution is given by $F(C_r^{S^1})$, and the element $U$ is the image of $1$ under the linear map $\mathbb{C}\rightarrow A$ corresponding to the M\"obius band as a morphism $\emptyset\rightarrow S^1$. For the invertible field theory defined by $F^{\mu_0}\defeq F^{\boldsymbol\mu}$ with $\mu_i=1$ for all $i\geq 1$ we observe the following.

		\begin{corollary}
			The Frobenius algebra corresponding to the invertible $\cat$-TFT defined by $F^{\mu_0}$ is $\mathbb{C}$ with its usual algebra structure, and pairing given by $\langle\boldsymbol x,\boldsymbol y\rangle=\mu_0^2\boldsymbol x\boldsymbol y$. The involutive automorphism is the identity, and the element $U$ is equal to $\mu_0^{-1}$.
			\qed
		\end{corollary}
%
\section{On classifying spaces}
%

	The cobordism hypothesis as treated in \cite {Lurie} is most naturally a statement about $(\infty, n)$-categories and their functors. In this setting higher homotopies play an important role. Thus, instead of factoring through the localisation, invertible functors are those that factor through the classifying space of the cobordism category appropriately interpreted, compare \cite {F}. Hence in that setting we need to understand the homotopy type of the classifying spaces in order to compute the invertible topological field theories. For the $(\infty, 1)$-category  $\EuScript{C}_2$  associated to $\cat$ this has been done in \cite {GTMW} where its classifying space is shown to be homotopic to the infinite loop space of a certain Thom spectrum. We consider here the classifying spaces of our discrete categories.

	\subsection{Classifying spaces and their fundamental groups}

		For any category $\mathscr{C}$ we denote its classifying space by $\text{B}\mathscr{C}$, recalling that this is by definition the geometric realisation of the nerve $\parallel N.\mathscr{C}\parallel$. If $\mathscr{C}$ is connected then the fundamental group $\pi_1(\text{B}\mathscr{C})$ is canonically isomorphic to the group of automorphisms of any object in the localisation $\mathscr{C}[\mathscr{C}^{-1}]$ \cite{Q}. Hence, having computed the localisation of a category, it is natural to ask what its classifying space might be.
		
		We briefly recall the following three useful facts regarding classifying spaces. Firstly, a functor $\text{F}\colon \mathscr{C}\rightarrow\mathscr{C}'$ induces a map $\text{BF}\colon \text{B}\mathscr{C}\rightarrow\text{B}\mathscr{C}'$. Secondly, for product categories we have $\text{B}(\mathscr{C}\times\mathscr{C}')\simeq\text{B}\mathscr{C}\times\text{B}\mathscr{C}'$. Finally, a natural transformation between two functors $\text{F}_0,\text{F}_1\colon\mathscr{C}\rightarrow\mathscr{C}'$ induces a homotopy $\text{BF}_0\rightarrow\text{BF}_1$ \cite{S}. It follows that a pair of adjoint functors $\text{F}_0\colon\mathscr{C}\rightarrow\mathscr{C}'$ and $\text{F}_1\colon\mathscr{C}'\rightarrow\mathscr{C}$ (and in particular an equivalence of categories) induces a homotopy equivalence $\text{B}\mathscr{C}\simeq\text{B}\mathscr{C}'$. 
 
		As before, let $\mathscr{G}(M)$ denote the group completion of a monoid $M$, that is the target of a universal homomorphism from $M$ to a group, and let $\Omega$ denote the loop space functor. The canonical homomorphism $M\rightarrow\mathscr{G}(M)$ induces a natural transformation $\Omega\text{B}(M)\rightarrow\mathscr{G}(M)$ which is a weak equivalence when $M$ satisfies certain properties as defined by Quillen \cite{QA}. In particular, this map is a weak equivalence when $M$ is abelian. It follows that $\text{B}(M)\simeq\text{B}\mathscr{G}(M)$ for abelian monoids $M$. Hence for a category $\mathscr{C}$ with abelian monoid structure we have $\text{B}\mathscr{C}\simeq\text{B}\mathscr{C}[\mathscr{C}^{-1}]$ induced by the canonical projection $\psi\colon\mathscr{C}\rightarrow\mathscr{C}[\mathscr{C}^{-1}]$.
		
		Recall that $B\mathbb{Z} \simeq S^1$ is the circle and $B\mathbb{Z}_2 \simeq \mathbb{R}P^\infty$ is the infinite real projective space. Our analysis of the subcategories in section 2 and Theorem 3.3  immediately give us the following.

		\begin{corollary} There are homotopy equivalences
			\begin{align*}
				& B \cat _0 = B\mathbb{N}^\infty \simeq T^\infty 	\\
				& B \cat ^+_1 = B (\mathbb{N} \times \mathbb {Z}_2) \simeq S^1 \times \mathbb{R}P^\infty 	\\
				& B \cat ^-_1 = B \mathbb{N} \simeq S^1	\\
				& B \catb \simeq B \mathbb{N} \simeq S^1 \\
				& B\cat _1 \simeq B \mathbb{Z} \simeq S^1.
			\end{align*}
			\qed
		\end{corollary}
		Our description of the categories lets us also describe the maps between the categories up to homotopy. Thus for example, the map $\text{B}\catoneor\rightarrow\text{B}\catone$ induced by the inclusion $\catoneor\hookrightarrow\catone$ is homotopic to $(z,x)\mapsto 2z$, while the map  $\text{B}\catonenor\rightarrow\text{B}\catone$ induced by the inclusion $\catonenor\hookrightarrow\catone$ is a homotopy equivalence.
  
		The computation of the localisations in section 3 immediately imply the following.

		\begin{corollary}
			The fundamental groups of classfying spaces can be identified as
			\[
				\pi_1(B\cat) \simeq 
				\pi_1(B\catoo) \simeq 
				\pi_1 (B\catoc) \simeq
				\pi_1 (B( \mathcal S \cap \catoo)) \simeq 
				\pi_1 (B\mathcal S) \simeq \mathbb Z 
			\]
			and 
			\[
				\pi_1 (B(\mathcal S \cap \cat)) \simeq \mathbb Z, \quad \quad
				\pi_1(B\catc) \simeq \mathbb{Z}\times \mathbb Z.
			\]
			The inclusion of $\mathcal S \cap \cat $ into \cat induces the multiplication by 2 map.
			\qed
		\end{corollary}

	\subsection{Discrete localisation conjecture}

		The computation of the localisations of the cobordism categories gives the first homotopy groups  of their classifying spaces. In the case of monoids, this completely determines the homotopy type of their classifying spaces,  as we have seen in Corollary 5.1. In those cases all higher homotopy groups are trivial. We conjecture that the same is true also for our main categories.
		
		\begin{conjecture}
			On classifying spaces the canonical maps from $\cat, \catoo, \catc, \catoc$  and the orientable subcategories to their respective localisations induce  homotopy equivalences. 
		\end{conjecture}

		In particular this would mean that  $\Theta$ induces on classifying spaces a homotopy equivalence with $S^1$ in the cases of \cat, \catoo , \catoc and their orientable subcategories, and with $S^1 \times S^1$ in case of \catc. 

		The analogue of the above conjecture for  the closed cobordism category for orientable surfaces $\mathcal {S} \cap \cat$ has been outstanding  many years, compare \cite{T}. It is especially surprising that these conjectures are not settled as the classifying spaces of the associated topological categories have been computed in \cite{GTMW}. 

		We note that up to homotopy, any CW-complex is the classifying space of a discrete category and any infinite loop space is the classifying space of a symmetric monoidal category. We offer the following result in partial support of the above conjecture.

		\begin{theorem}
		\label{monoidalcat}
			Let $\mathcal C = \cat, \catoo,  \catoc$ or one of their orientable subcategories. Then there exists a simply connected infinite loop space $X$ (depending on $\mathcal C$) such that $B\mathcal C$ is homotopy equivalent to $X\times S^1$. Similarly, there exists a simply connected infinite loop space $X$ such that $B\catc$ is homotopy equivalent to $X \times S^1 \times S^1$.
		\end{theorem}
		
		\begin{proof}
			We will only discuss the case $\mathcal C = \cat$. The other cases are similar. See also \cite {T}.
			
			It is well-known that the classifying space of a connected symmetric strict monoidal category has the homotopy type of an infinite loop space \cite{MA}, \cite{S1}. Moreover, a symmetric monoidal functor induces a map of infinite loop spaces on the classifying spaces of two such categories. Since the Euler characteristic is additive with respect to disjoint union, it follows that the induced map $\Theta\colon\text{B}\cat\rightarrow\text{B}\mathbb{Z}$ is a map of infinite loop spaces.
						
			The composition $\frac{\mathbb{N}\times\mathbb{Z}_2}{\sim}=\catonenor\hookrightarrow\cat\overset{\Theta}{\longrightarrow}\mathbb{Z}$ is $(n,\epsilon)\mapsto -n$, in other words it is the canonical map  $\frac{\mathbb{N}\times\mathbb{Z}_2}{\sim}\rightarrow\mathscr{G}\left(\frac{\mathbb{N}\times\mathbb{Z}_2}{\sim}\right)=\mathbb{Z}$ composed with an isomorphism. Therefore the induced map on classifying spaces is a homotopy equivalence. It follows that $\Theta$ has a section, that is a map $s\colon\text{B}\mathbb{Z}\rightarrow\text{B}\cat$ with $\Theta\circ s\simeq\text{Id}_{\text{B}\mathbb{Z}}$.
			
			Let $X$ be the homotopy fibre of $\Theta\colon\text{B}\cat\rightarrow\text{B}\mathbb{Z}$, and denote the map $X\rightarrow\text{B}\cat$ by $\phi$. Note that $X$ is an infinite loop space, since $\Theta$ is a map of infinite loop spaces. We show that $X$ is simply connected, and that the map $\Psi\colon X\times\text{B}\mathbb{Z}\rightarrow\text{B}\cat$ defined by $(x,y)\mapsto\phi(x)\cdot s(y)$ is a weak equivalence, where the product $\cdot$ on $\text{B}\cat$ is the loop product.

			 Since $\Theta$ has a section, the long exact sequence in homotopy
			\[
				\cdots \rightarrow \pi_nX \overset{\phi_*}{\rightarrow} \pi_n(\text{B}\cat) \overset{\Theta_*}{\rightarrow} \pi_nS^1 \rightarrow \pi_{n-1}X \rightarrow \cdots
			\]
			splits as short exact sequences $0\rightarrow\pi_nX\overset{\phi_*}{\rightarrow}\pi_n(\text{B}\cat)\overset{\Theta_*}{\rightarrow}\pi_nS^1\rightarrow 0$. Considering this short exact sequence in dimensions $n=0$ and $n=1$, we see that $X$ is simply connected. We therefore have split short exact sequences of abelian groups in every dimension, and hence isomorphisms $\pi_n(X\times S^1)\rightarrow\pi_n(\text{B}\cat)$ given by $[x,y]\mapsto\phi_*[x]s_*[y]$. But these are precisely the maps induced by $\Psi$. Hence $\Psi$ is a weak equivalence, and therefore a homotopy equivalence by Whitehead's theorem as required.
		\end{proof}
%
%
\appendix
\section{The category of based symmetric monoidal functors}

	In this appendix we introduce the notions of based symmetric monoidal categories, based functors and based equivalences between them. We prove that the associated functor categories are equivalent if the souce or target category is replaced by an equivalent category. Furthermore, this remains true when considering invertible functors. To define the latter the Picard category of a symmetric monoidal category is introduced. 

	\subsection{Based symmetric monoidal equivalences}
	\label{appendix 1}
 
		{\bf Notation}. $\langle\mathscr{C},\otimes,e,\alpha,\lambda,\rho,\gamma\rangle$ will denote the symmetric monoidal category $\mathscr{C}$ with product $\otimes\colon \mathscr{C}\times\mathscr{C}\rightarrow\mathscr{C}$, unit $e\in\mathscr{C}$, associator $\alpha$, left and right unitors $\lambda$ and $\rho$, and braiding $\gamma$. Recall that for all objects $a,b,c\in\mathscr{C}$ the components
		\begin{align*}
			\alpha_{a,b}\colon a\otimes(b\otimes c)&\rightarrow(a\otimes b)\otimes c\\
			\lambda_a\colon e\otimes a&\rightarrow a\\
			\rho_a\colon a\otimes e&\rightarrow a\\
			\gamma_{a,b}\colon a\otimes b&\rightarrow b\otimes a
		\end{align*}
		are isomorphisms in $\mathscr{C}$, natural in $a,b$ and $c$, and must satisfy certain commutativity requirements as described by Mac Lane \cite{MAC}. Recall also that
		\[
			\gamma_{b,a}\circ\gamma_{a,b}=1_{a\otimes b}.
		\]
		Finally, recall that $\langle\mathscr{C},\otimes,e,\alpha,\lambda,\rho,\gamma\rangle$ is a symmetric \emph{strict} monoidal category if the morphisms $\alpha, \rho$ and $\lambda$ are all identities.

		We shall often denote the symmetric monoidal category $\langle\mathscr{C},\otimes,e,\alpha,\lambda,\rho,\gamma\rangle$ simply by $\mathscr{C}$.

		\begin{definition}
			Let $\mathscr{C}$ and $\mathscr{C'}$ be symmetric monoidal categories with respective units $e$ and $e'$. Then we say that
			\begin{enumerate}[(i)]
				\item a functor $F\colon \mathscr{C}\rightarrow\mathscr{C'}$ is \emph{based} if $F(e)=e'$;
				\item a natural transformation $\tau\colon F\rightarrow G$ between two based functors $F,G\colon \mathscr{C}\rightarrow\mathscr{C'}$ is \emph{based} if $\tau_e=1_{e'}$.
			\end{enumerate}
		\end{definition}
			
		\begin{definition}
		\label{def: based symmetric monoidal functor}
			Let $\mathscr{C}$ and $\mathscr{C}'$ be symmetric monoidal categories. We define a \emph{based symmetric monoidal functor}\footnote{In the terminology of Mac Lane \cite{MAC}, $(F,F_2)$ is a symmetric \emph{strong} monoidal functor which is strict with respect to units.} $(F,F_2)\colon \mathscr{C}\rightarrow\mathscr{C}'$ to be a pair consisting of
			\begin{enumerate}[(i)]
				\item an (ordinary) based functor $F\colon \mathscr{C}\rightarrow\mathscr{C}'$;
				\item for each pair of objects $a,b\in\mathscr{C}$ an isomorphism
				\[
					F_2(a,b)\colon F(a)\otimes'F(b)\rightarrow F(a\otimes b)
				\]
				in $\mathscr{C}'$ which is natural in $a$ and $b$.
			\end{enumerate}
			Together these must make the following four diagrams commute in $\mathscr{C}'$.
			\[
				\begin{array}{c}
					\begin{aligned}
						\xymatrix{
							F(a) \otimes'(F(b) \otimes'F(c))\ar[r]^{\alpha'}\ar[d]_{1\otimes'F_2} & (F(a) \otimes'F(b)) \otimes'F(c)\ar[d]^{F_2\otimes'1} \\
							F(a) \otimes'F(b\otimes c)\ar[d]_{F_2} & F(a\otimes b) \otimes'F(c)\ar[d]^{F_2} \\
							F(a \otimes(b \otimes c))\ar[r]_{F(\alpha)} & F((a \otimes b)\otimes c)
						}
					\end{aligned}
					\\ \\
					\begin{array}{ccc}
						\begin{aligned}
							\xymatrix{
								F(a)\otimes'e'\ar[r]^{\rho'}\ar@{=}[d] & F(a) \\
								F(a)\otimes'F(e)\ar[r]_{F_2}& F(a\otimes e)\ar[u]_{F(\rho)}
							}
						\end{aligned}
						& &
						\begin{aligned}
							\xymatrix{
								 e'\otimes'F(a)\ar[r]^{\lambda'}\ar@{=}[d] & F(a)  \\
								 F(e)\otimes'F(a)\ar[r]_{F_2} & F(e\otimes a)\ar[u]_{F(\lambda)}
							} 
						\end{aligned}
					\end{array}
					\\ \\
					\begin{aligned}
						\xymatrix{
							F(a)\otimes' F(b)\ar[r]^{\gamma'}\ar[d]_{F_2} & F(b)\otimes'F(a)\ar[d]^{F_2} \\
         					F(a\otimes b)\ar[r]_{F(\gamma)} & F(b\otimes a)	
						}
					\end{aligned}
				\end{array}
			\]
		\end{definition}
		In the case of a strict source and target, we say that $(F,F_2)$ is a based symmetric \emph{strict} monoidal functor if all $F_2$ are identities.
%
		\begin{lemma}
		\label{lemma: composition of monoidal functors}
			Let $(F,F_2)\colon \mathscr{C}\rightarrow\mathscr{C}'$ and $(G,G_2)\colon \mathscr{C}'\rightarrow\mathscr{C}''$ be based symmetric monoidal functors. Then the composite $(G\circ F,(G\circ F)_2)\colon \mathscr{C}\rightarrow\mathscr{C}''$ is a based symmetric monoidal functor, where we define
			\[
				(G\circ F)_2(a,b)\defeq G(F_2(a,b))\circ G_2(F(a),F(b))
			\]
			for all objects $a,b\in\mathscr{C}$.\qed 
		\end{lemma}
%
		\begin{definition}
		\label{def: based monoidal natural transformation}
			Let $\mathscr{C},\mathscr{C}'$ be symmetric monoidal categories, and let $(F,F_2),(G,G_2)\colon \mathscr{C}\rightarrow\mathscr{C}'$ be based symmetric monoidal functors. A \emph{based monoidal natural transformation} $\tau\colon (F,F_2)\rightarrow (G,G_2)$ between $(F,F_2)$ and $(G,G_2)$ is a based natural transformation $\tau\colon F\rightarrow G$ between the underlying ordinary functors, such that the following diagram commutes in $\mathscr{C'}$ for all objects $a,b\in\mathscr{C}$.
			\[
				\xymatrix{
					F(a)\otimes'F(b)\ar[r]^{F_2}\ar[d]_{\tau_a\otimes'\tau_b} & F(a\otimes b)\ar[d]^{\tau_{a\otimes b}} \\
					G(a)\otimes'G(b)\ar[r]_{G_2} & G(a\otimes b)
				}
			\]
		\end{definition}
		\begin{remark}
			One could ask that $\tau$ be in some sense symmetric, in that it should satisfy $(\tau_b\otimes'\tau_a)\circ\gamma'_{F(a),F(b)}=\gamma'_{G(a),G(b)}\circ(\tau_a\otimes'\tau_b)$ and $\tau_{b\otimes a}\circ F(\gamma_{a,b})=G(\gamma_{a,b})\circ\tau_{a\otimes b}$. But this is redundant; both properties hold automatically by naturality of $\gamma'$ and $\tau$.
		\end{remark}
%
		\begin{lemma}
		\label{lemma: composition of based monoidal natural transformations}
			Let $(F,F_2),(G,G_2),(H,H_2)\colon \mathscr{C}\rightarrow\mathscr{C}'$ be based symmetric monoidal functors between symmetric monoidal categories. Let $\tau\colon (F,F_2)\rightarrow (G,G_2)$ and $\sigma\colon (G,G_2)\rightarrow (H,H_2)$ be based monoidal natural transformations. Then the composite $\sigma\circ\tau\colon (F,F_2)\rightarrow (H,H_2)$ is a based monoidal natural transformation, where $(\sigma\circ\tau)_a\defeq\sigma_a\circ\tau_a$ for each object $a\in\mathscr{C}$.
			\qed
 		\end{lemma}
%
		\begin{definition}
			For symmetric monoidal categories $\mathscr{C}$ and $\mathscr{D}$, denote by \linebreak $\operatorname{SymmMon}[\mathscr{C},\mathscr{D}]_*$ the category whose objects are based symmetric monoidal functors $\mathscr{C}\rightarrow\mathscr{D}$, and whose morphisms are based monoidal natural transformations. 
		\end{definition}
%
		\begin{definition}
		\label{def: symmetric monoidal equivalence}
			Let $\mathscr{D}$ and $\mathscr{D}'$ be symmetric monoidal categories. A \emph{based symmetric monoidal equivalence} between $\mathscr{D}$ and $\mathscr{D}'$ is a quadruple $((\phi,\phi_2),(\psi,\psi_2),\eta,\varepsilon)$ where
			\begin{enumerate}[(i)]
				\item $(\phi,\phi_2)\colon \mathscr{D}\rightarrow\mathscr{D}'$ and $(\psi,\psi_2)\colon \mathscr{D}'\rightarrow\mathscr{D}$ are based symmetric monoidal functors;
				\item $\phi\colon \mathscr{D}\rightarrow\mathscr{D}'$ and $\psi\colon \mathscr{D}'\rightarrow\mathscr{D}$ form an equivalence of categories in the usual way;
				\item the unit $\eta\colon 1_{\mathscr{D}}\rightarrow\psi\circ\phi$ and counit $\varepsilon\colon \phi\circ\psi\rightarrow 1_{\mathscr{D}'}$ of this equivalence are based monoidal natural isomorphisms, i.e.\ based monoidal natural  transformations with $\eta_a$ and $\varepsilon_a$ isomorphisms for all objects $a$.
			\end{enumerate}
		\end{definition}
		We often say that a single functor $(\phi,\phi_2)$ is a based symmetric monoidal equivalence if the corresponding $(\psi,\psi_2),\eta$ and $\varepsilon$ exist.
%
		\begin{proposition}
		\label{prop: equivalence of categories induces equivalence of functor categories}
			Let $\mathscr{C}$ and $\mathscr{C'}$ as well as $\mathscr{D}$ and $\mathscr{D'}$ be pairs of equivalent based symmetric monoidal categories. 
Then the functor categories $\operatorname{SymmMon}[\mathscr{C},\mathscr{D}]_*$ and $\operatorname{SymmMon}[\mathscr{C'},\mathscr{D}']_*$ are equivalent.
		\end{proposition}
		\begin{proof}
			Suppose we have a based symmetric monoidal equivalence $((\phi,\phi_2),(\psi,\psi_2),\eta,\varepsilon)$ between $\mathscr{D}$ and $\mathscr{D'}$ as in Definition~\ref{def: symmetric monoidal equivalence}. 
			For objects $(F,F_2),(G,G_2)$ and morphisms $\tau\colon (F,F_2)\rightarrow (G,G_2)$ in $\operatorname{SymmMon}[\mathscr{C},\mathscr{D}]_* $ we define a functor $\Phi\colon \operatorname{SymmMon}[\mathscr{C},\mathscr{D}]_*\rightarrow\operatorname{SymmMon}[\mathscr{C},\mathscr{D}']_*$ as follows.
			\begin{align*}
				\Phi\colon \operatorname{SymmMon}[\mathscr{C},\mathscr{D}]_*  \rightarrow & \operatorname{SymmMon}[\mathscr{C},\mathscr{D}']_* \\
				(F,F_2) \mapsto & (\phi\circ F,(\phi\circ F)_2) \\
				\tau \mapsto & \Phi(\tau) \mbox{\ where\ } \Phi(\tau)_a\defeq\phi(\tau_a)\mbox{\ for all objects $a\in\mathscr{C}$.}
			\end{align*}
			It is routine to check that $\Phi$ is functorial, and by Lemma \ref{lemma: composition of monoidal functors} $\Phi((F,F_2))$ is indeed an object in $\operatorname{SymmMon}[\mathscr{C},\mathscr{D}']_*$. Naturality of $\Phi(\tau)$ follows from naturality of $\tau$, and $\Phi(\tau)$ is based since $\phi$ and $\tau$ are based. 
			
			To check that $\Phi(\tau)$ is monoidal one checks the commutativity of the following diagram. For brevity we denote the tensor products in $\mathscr{C},\mathscr{D}$ and $\mathscr{D}'$ all by $\otimes$.
			\[
				\xymatrixcolsep{5pc}
				\xymatrix{
					(\phi\circ F)(a)\otimes(\phi\circ F)(b)\ar[r]^-{(\phi\circ F)_2}\ar[d]_{\phi(\tau_a)\otimes\phi(\tau_b)} & (\phi\circ F)(a\otimes b)\ar[d]^{\phi(\tau_{a\otimes b})} \\
					(\phi\circ G)(a)\otimes(\phi\circ G)(b)\ar[r]_-{(\phi\circ G)_2} & (\phi\circ G)(a\otimes b)	
				}
			\]
			In the other direction we define a functor
			\begin{align*}
				\Psi\colon \operatorname{SymmMon}[\mathscr{C},\mathscr{D}']_*  \rightarrow & \operatorname{SymmMon}[\mathscr{C},\mathscr{D}]_* \\
				(F,F_2) \mapsto & (\psi\circ F,(\psi\circ F)_2) \\
				\tau \mapsto & \Psi(\tau)\mbox{ where } \Psi(\tau)_a\defeq\psi(\tau_a)\mbox{ for all objects $a\in\mathscr{C}$}
			\end{align*}
			and by similar arguments this satisfies all the required properties. It remains to show that $\Phi$ and $\Psi$ define an equivalence of categories. Define a natural isomorphism
			\[
				\bar{\eta}\colon 1_{\operatorname{SymmMon}[\mathscr{C},\mathscr{D}]_* }\rightarrow\Psi\circ\Phi
			\]
			by
			\[
				({\bar{\eta}_{(F,F_2)}})_a\defeq\eta_{F(a)}
			\]
			for objects $a\in\mathscr{C}$. We have two things to check: that the components $\bar{\eta}_{(F,F_2)}$ are indeed isomorphisms in $\operatorname{SymmMon}[\mathscr{C},\mathscr{D}]_*$, and that $\bar{\eta}$ is natural.

			To show that $\bar{\eta}_{(F,F_2)}\colon (F,F_2)\rightarrow(\psi\circ\phi\circ F,(\psi\circ\phi\circ F)_2)$ is a based monoidal natural isomorphism we argue as follows. Clearly $({\bar{\eta}_{(F,F_2)}})_a$ is invertible, since $\eta$ is a natural isomorphism. We also see that $\bar{\eta}_{(F,F_2)}$ is based since $F$ and $\eta$ are based, and naturality follows from naturality of $\eta$. It remains to check that $\bar{\eta}_{(F,F_2)}$ is monoidal, for which we must verify the commutativity of the following diagram.
			\[
				\xymatrixcolsep{8pc}
				\xymatrix{
					F(a)\otimes F(b)\ar[r]^{F_2}\ar[d]_{\eta_{F(a)}\otimes\eta_{F(b)}} & F(a\otimes b)\ar[d]^{\eta_{F(a\otimes b)}} \\
					(\psi\circ\phi\circ F)(a)\otimes(\psi\circ\phi\circ F)(b)\ar[r]_-{(\psi\circ\phi\circ F)_2} & (\psi\circ\phi\circ F)(a\otimes b)
				}
			\]
			
			Finally, one checks that naturality of $\bar{\eta}$ again follows from naturality of $\eta$. Conversely we define
			\[
				\bar{\varepsilon}\colon \Phi\circ\Psi\rightarrow 1_{\operatorname{SymmMon}[\mathscr{C},\mathscr{D}']_*}
			\]
			by
			\[
				(\bar{\varepsilon}_{(F,F_2)})_a\defeq\varepsilon_{F(a)}.
			\]
			Via similar arguments we see that $\bar{\varepsilon}$ is a natural isomorphism, and hence $(\Phi,\Psi,\bar{\eta},\bar{\varepsilon})$ defines an equivalence of categories. One proceeds in a similar fashion when replacing $\mathscr C$ by $\mathscr {C'}$.
		\end{proof}
%
%
	\subsection{The skeleton of a symmetric monoidal category}
		The purpose of this subsection is to show that the skeleton of a symmetric monoidal category has a symmetric monoidal structure such that it is equivalent to the category itself as a symmetric monoidal category. 
		\begin{proposition}
			\label{prop: skeleton is symmetric monoidal}
			Let $\langle \mathscr{C},\otimes,e,\alpha,\lambda,\rho,\gamma\rangle$ be a symmetric monoidal category, and let $\mathscr{C}^s\subseteq\mathscr{C}$ be a skeleton of $\mathscr{C}$. Denote by $\phi(a)$ the object in $\mathscr{C}^s$ representing the isomorphism class of the object $a\in\mathscr{C}$, and  set $\phi(e)\defeq e$. Choose one fixed isomorphism $\phi_a\colon a\rightarrow\phi(a)$ for each $a$, setting $\phi_{\phi(a)}\defeq 1_{\phi(a)}$. Then $\mathscr{C}^s$ can be given a monoidal structure $\langle \mathscr{C}^s,\otimes^s,e,\alpha^s,\lambda^s,\rho^s,\gamma^s\rangle$ where the structure maps are defined as follows. For objects $u,u',v,v',w$ and morphisms $f\colon u\rightarrow v,f'\colon u'\rightarrow v'$ in $\mathscr{C}^s$ we define
			\begin{eqnarray*}
				\mbox{ {\bf Product} } & u\otimes^s u'&  \defeq\phi(u\otimes u') \\
				& f\otimes^sf' &\defeq\phi_{v\otimes v'}\circ(f\otimes f')\circ\phi_{u\otimes u'}^{-1} \\
				\mbox{ {\bf Associator} } &  \alpha^s_{u,v,w}  &\defeq\phi_{\phi(u\otimes v)\otimes w}\circ(\phi_{u\otimes v}\otimes 1_w)\circ\alpha_{u,v,w}\circ(1_u\otimes\phi_{v\otimes w}^{-1})\circ\phi_{u\otimes\phi(v\otimes w)}^{-1} \\
				\mbox{ {\bf Unitors} } & \lambda^s_u &\defeq\lambda_u\circ\phi_{e\otimes u}^{-1} \\
				& \rho^s_u & \defeq\rho_u\circ\phi_{u\otimes e}^{-1} \\
				\mbox{ {\bf Braiding} } & \gamma_{u,v}&\defeq\phi_{v\otimes u}\circ\gamma_{u,v}\circ\phi_{u\otimes v}^{-1}.
			\end{eqnarray*}
		\end{proposition}
		\begin{proof}
		Checking that the appropriate diagrams commute is a long but standard process. The required properties hold for the skeletal structure maps in $\mathscr{C}^s$ as a consequence of the fact that they hold for the analogous maps in $\mathscr{C}$.
		\end{proof}
%
		\begin{proposition}
		\label{prop: the inclusion of a skeleton is an equivalence}
			The inclusion $(i,i_2)\colon \mathscr{C}^s\hookrightarrow\mathscr{C}$ where $i_2(u,v)\defeq\phi_{u\otimes v}$ is a based symmetric monoidal equivalence of categories. \qed
		\end{proposition}
	Putting together Proposition~\ref{prop: equivalence of categories induces equivalence of functor categories} and Proposition~\ref{prop: the inclusion of a skeleton is an equivalence} we obtain the following.
%
		\begin{corollary}
		\label{corr: the inclusion of a skeleton induces an equivalence of functor categories}
			Let $\mathscr{C}$ and $\mathscr{D}$ be symmetric monoidal categories, and let $\mathscr{D}^s$ be a skeleton of $\mathscr{D}$ with symmetric monoidal structure as in Proposition \ref{prop: skeleton is symmetric monoidal}. Then the inclusion $\mathscr{D}^s\hookrightarrow\mathscr{D}$ induces an equivalence of categories $\operatorname{SymmMon}[\mathscr{C},\mathscr{D}^s]_*\rightarrow\operatorname{SymmMon}[\mathscr{C},\mathscr{D}]_*$. \qed
		\end{corollary}
		\subsubsection{A skeleton for \vect}
        \label{a skeleton for vect}
%
			Define $\vects\subset\vect$ to be the skeleton of \vect whose objects are $\{\mathbb{C}^n|n\in\mathbb{N}_{>0}\}$ (so for an $n$-dimensional vector space $V^n$ we have $\phi(V^n)\defeq\mathbb{C}^n$). Choose the isomorphism $\phi_{\mathbb{C}\otimes\mathbb{C}}:\mathbb{C}\otimes\mathbb{C}\rightarrow\mathbb{C}$ in the canonical way, that is to say $z\otimes z'\mapsto zz'$. We give \vects a symmetric monoidal structure as in Proposition~\ref{prop: skeleton is symmetric monoidal}. Explicitly this means that for linear maps $\sigma\colon \mathbb{C}^n\rightarrow\mathbb{C}^m$ and $\sigma'\colon \mathbb{C}^{n'}\rightarrow\mathbb{C}^{m'}$ in \vects we have
			\begin{align*}
				\sigma\otimes^s\sigma' & \defeq \phi_{\mathbb{C}^m \otimes\mathbb{C}^{m'}}\circ(\sigma \otimes\sigma') \circ\phi_{\mathbb{C}^n \otimes\mathbb{C}^{n'}}^{-1}\colon \mathbb{C}^{nn'}\rightarrow\mathbb{C}^{mm'}.
			\end{align*}
%
			In particular if we consider two linear maps $\sigma,\sigma'\colon \mathbb{C}\rightarrow\mathbb{C}$, that is to say two elements $\sigma,\sigma'\in\mathbb{C}$, the canonical choice of $\phi_{\mathbb{C}\otimes\mathbb{C}}$ has two consequences. Firstly we have $\sigma\otimes^s\sigma'=\sigma\sigma'\in\mathbb{C}$. Secondly, the structure map components $\alpha_{\mathbb{C},\mathbb{C},\mathbb{C}}^s$, $\lambda_{\mathbb{C}}^s,\rho_{\mathbb{C}}^s$ and $\gamma_{\mathbb{C},\mathbb{C}}^s$ are all identities, and so the monoidal structure on $\mathbb{C}\subset\vects$ is strict. 

	\subsection{Invertible functors of symmetric monoidal categories}

		An invertible functor $F\colon \mathscr C \rightarrow \mathscr D$ sends every morphism in $\mathscr C$ to an invertible morphism in $\mathscr D$. In particular, the functor factors through the localisation $\mathscr C [\mathscr C ^{-1}]$. In the presence of a symmetic monoidal structure we consider a more refined notion of invertible functors for which the image of any object is also invertible.

		\begin{definition}
			An object $a$ of a symmetric monoidal category $\mathscr D$ is called \emph{invertible} if there exists another object $\bar a$ with $a \otimes \bar a$ isomorphic to the unit $e$.	The Picard category $Pic (\mathscr D)$  is the  subcategory of invertible objects and morphisms. 
		\end{definition}

		For example, the Picard category of  $\text{Vect}_{\mathbb{C}}$ is the category of complex lines, and the Picard group of its skeleton is $\mathbb{C}^\times$, the category of one object and morphisms the group of non-zero complex numbers under multiplication. 

		\begin {definition}
			An invertible functor  of  symmetric monoidal categories is a functor  $F\colon \mathscr C \rightarrow Pic (\mathscr D)$. We define the category of based invertible functors as
			\[
				\operatorname{SymmMon}^\times[\mathscr{C},\mathscr{D}]_* \defeq \operatorname{SymmMon}[\mathscr{C}, Pic (\mathscr{D})]_*.
			\]
		\end{definition}
		 
		For any symmetric monoidal category $\mathscr{C}$, we can extend the product on $\mathscr{C}$ to a product on its localisation $\mathscr{C}[\mathscr{C}^{-1}]$ by defining $f^{-1}\otimes g^{-1}\defeq (f\otimes g)^{-1}$ for inverses and extending bifunctorially. The localised category then has a symmetric monoidal structure inherited from $\mathscr{C}$ by transferring the unit and structure maps along the canonical projection $\mathscr{C}\rightarrow\mathscr{C}[\mathscr{C}^{-1}]$. Furthermore, for a second symmetric monoidal category $\mathscr{D}$, a morphism inverting functor $F\colon\mathscr{C}\rightarrow\mathscr{D}$ is based symmetric monoidal if and only if the corresponding functor $\bar{F}\colon\mathscr{C}[\mathscr{C}^{-1}]\rightarrow\mathscr{D}$ is based symmetric monoidal. Hence
		\[
			\operatorname{SymmMon}^\times[\mathscr{C},\mathscr{D}]_*  = \operatorname{SymmMon}[\mathscr{C} [\mathscr {C} ^{-1}],
			Pic (\mathscr{D})]_*.
		\]
		We derive a version of Corollary~\ref{corr: the inclusion of a skeleton induces an equivalence of functor categories} in this setting.
%
		\begin{corollary}
		\label{corr: skeletal inclusion invertible version}
			Let $\mathscr{C}$ and $\mathscr{D}$ be symmetric monoidal categories, and let $\mathscr{D}^s$ be a skeleton of $\mathscr{D}$. Then the inclusion $\mathscr{D}^s\hookrightarrow\mathscr{D}$ induces an equivalence of categories 
		\[
			\operatorname{SymmMon}^\times [\mathscr{C},\mathscr{D}^s]_* \longrightarrow \operatorname{SymmMon}^\times [\mathscr{C},\mathscr{D}]_*.
		\] 
		\end{corollary}

		\begin{proof}
			It is straight forward to show that if $\mathscr D$ and $\mathscr {D'}$ are equivalent based symmetric monoidal categories then so are their Picard categories.
		\end{proof}

%
%
%

\providecommand{\bysame}{\leavevmode\hbox to3em{\hrulefill}\thinspace}
\begin{thebibliography}{GTMW09}

		\bibitem[Abr96]{Ab}
		L.~Abrams, \emph{Two-dimensional topological quantum field theories and
		  {F}robenius algebras}, J. Knot Theory Ramifications \textbf{5} (1996), no.~5,
		  569--587.

		\bibitem[AN06]{AN}
		A.~Alexeevski and S.~Natanzon, \emph{Noncommutative two-dimensional topological
		  field theories and {H}urwitz numbers for real algebraic curves}, Selecta
		  Math. (N.S.) \textbf{12} (2006), no.~3-4, 307--377.

		\bibitem[Ati89]{A}
		M.~Atiyah, \emph{Topological quantum field theories}, Inst. Hautes {\'E}tudes
		  Sci. Publ. Math. \textbf{68} (1989), 175--186.

		\bibitem[Bra12]{B}
		C.~Braun, \emph{Moduli spaces of {K}lein surfaces and related operads}, Algebr.
		  Geom. Topol. \textbf{12} (2012), no.~3, 1831–1899.

		\bibitem[Dou01]{D}
		C.~Douglas, \emph{Notes on the open $(1+1)$-dimensional cobordism category},
		  in: MSc (Res) Thesis, University of Oxford, 2001.

		\bibitem[Fre13]{F}
		D.~S. Freed, \emph{The cobordism hypothesis}, Bull. Amer. Math. Soc.
		  \textbf{50} (2013), 57--92.

		\bibitem[GTMW09]{GTMW}
		S.~Galatius, U.~Tillmann, I.~Madsen, and M.~Weiss, \emph{The homotopy type of
		  the cobordism category}, Acta Math. \textbf{202} (2009), no.~2, 195--239.

		\bibitem[GZ67]{GZ}
		P.~Gabriel and M.~Zisman, \emph{Calculus of fractions and homotopy theory},
		  Springer-Verlag, 1967.

		\bibitem[Koc04]{K}
		J.~Kock, \emph{Frobenius algebras and 2{D} topological quantum field theories},
		  London Mathematical Society Student Texts, vol.~59, Cambridge University
		  Press, 2004.

		\bibitem[LP08]{LP}
		A.~D. Lauda and H.~Pfeiffer, \emph{Open-closed strings: two-dimensional
		  extended {TQFT}s and {F}robenius algebras}, Topology Appl. \textbf{155}
		  (2008), no.~7, 623--666.

		\bibitem[Lur09]{Lurie}
		J.~Lurie, \emph{On the classification of topological field theories}, Current
		  developments in mathematics, 2008, Int. Press, Somerville, MA, 2009,
		  pp.~129--280.

		\bibitem[May74]{MA}
		J.~P. May, \emph{{$E_{\infty }$} spaces, group completions, and permutative
		  categories}, in: New developments in topology, London Math. Soc. Lecture Note
		  Ser., vol.~11, Cambridge University Press, 1974, pp.~61--93.

		\bibitem[ML98]{MAC}
		S.~Mac~Lane, \emph{Categories for the working mathematician}, second ed.,
		  Graduate Texts in Mathematics, vol.~5, Springer-Verlag, 1998.

		\bibitem[MS06]{MS}
		G.~W. Moore and G.~Segal, \emph{{D}-branes and {K}-theory in 2{D} topological
		  field theory}, ar{X}iv:hep-th/0609042v1 (2006).

		\bibitem[Qui73]{Q}
		D.~Quillen, \emph{Higher algebraic {$K$}-theory: {I}}, in: Algebraic
		  {$K$}-theory {I}, Springer Lect. Notes Math., vol. 341, Springer-Verlag,
		  1973, pp.~77--139.

		\bibitem[Qui94]{QA}
		\bysame, \emph{On the group completion of a simplicial monoid}, appendix to:
		  Filtrations on the homology of algebraic varieties (E. M. Friedlander and B.
		  Mazur), Mem. Amer. Math. Soc. \textbf{110} (1994), no.~529.

		\bibitem[Seg68]{S}
		G.~Segal, \emph{Classifying spaces and spectral sequences}, Inst. Hautes
		  {\'E}tudes Sci. Publ. Math. \textbf{34} (1968), 105--112.

		\bibitem[Seg74]{S1}
		\bysame, \emph{Categories and cohomology theories}, Topology \textbf{13}
		  (1974), 293--312.

		\bibitem[Til96]{T}
		U.~Tillmann, \emph{The classifying space of the {$1+1$} dimensional cobordism
		  category}, J. Reine Angew. Math. \textbf{479} (1996), 67--75.

		\bibitem[TT06]{TT}
		V.~Turaev and P.~Turner, \emph{Unoriented topological quantum field theory and
		  link homology}, Algebr. Geom. Topol. \textbf{6} (2006), 1069--1093.

		\bibitem[Wit82]{W}
		E.~Witten, \emph{Supersymmetry and {M}orse theory}, J. Differential Geom.
		  \textbf{17} (1982), no.~4, 661--692.

	\end{thebibliography}
	\providecommand{\bysame}{\leavevmode\hbox to3em{\hrulefill}\thinspace}
	
%

\end{document}